\documentclass[british,a4paper,10pt]{article}
\usepackage{babel,amsmath,amssymb,amsthm}
\usepackage{bm}
\usepackage[sans]{dsfont}
\usepackage[mathscr]{euscript}
\usepackage{times}
\usepackage{cite}
\usepackage{xcolor}
\usepackage{mathtools}
\usepackage[all]{xy}
\usepackage{float}
\usepackage{hyperref}
\usepackage{a4wide}
\theoremstyle{plain}
\newtheorem{theorem}{Theorem}
\newtheorem{proposition}[theorem]{Proposition}

\theoremstyle{remark}
\newtheorem{remark}{Remark}

\theoremstyle{definition}
\newtheorem{definition}{Definition}
\theoremstyle{definition}

\newcommand{\Tr}{\operatorname{Tr}}
\newcommand{\rmd}{\mathrm{d}}
\newcommand{\rme}{\mathrm{e}}
\newcommand{\rmi}{\mathrm{i}}
\newcommand{\RE}{\mathrm{\,Re\,}}

\newcommand{\Ebb}{\mathbb{E}}
\newcommand{\Rbb}{\mathbb{R}}
\newcommand{\Cbb}{\mathbb{C}}
\newcommand{\Sbb}{\mathbb{S}}

\newcommand{\T}{\mathtt{T}}

\newcommand{\openone}{\mathds{1}}
\let\idty\openone
\newcommand{\ind}{\mathtt{1}}
\newcommand{\id}{\mathrm{Id}}

\newcommand{\norm}[1]{\left\Vert#1\right\Vert}
\newcommand{\abs}[1]{\left\vert#1\right\vert}

\newcommand{\Bcal}{\mathcal{B}}

\newcommand{\Hcal}{\mathcal{H}}
\newcommand{\Ical}{\mathcal{I}}

\newcommand{\Kcal}{\mathcal{K}}
\newcommand{\Lcal}{\mathcal{L}}

 \newcommand{\Rcal}{\mathcal{R}}

\newcommand{\Tcal}{\mathcal{T}}

 \newcommand{\Wcal}{\mathcal{W}}

\newcommand{\Bscr}{\mathscr{B}}

\newcommand{\Hscr}{\mathscr{H}}

\newcommand{\Kscr}{\mathscr{K}}

 \newcommand{\Nscr}{\mathscr{N}}

\newcommand{\Tscr}{\mathscr{T}}

\def\Order{{\mathbf O}}

\def\LevKhi{L\'evy-Khintchine}

\begin{document}
\title{Hybrid quantum-classical systems: Quasi-free Markovian dynamics}
\author{Alberto Barchielli\footnote{Istituto Nazionale di Fisica Nucleare (INFN), Sezione di Milano,Italy; also
Istituto Nazionale di Alta Matematica (INDAM-GNAMPA)}, \ Reinhard F. Werner\footnote{Institut f\"ur Theoretische Physik, Leibniz Universit\"at, Hannover, Germany}}

\maketitle

\begin{abstract}
In the case of a quantum-classical hybrid system with a finite number of degrees of freedom, the problem of characterizing the most general dynamical semigroup is solved, under the restriction of being \emph{quasi-free}. This is a generalization of a Gaussian dynamics, and it is defined by the property of sending (hybrid) Weyl operators into Weyl operators in the Heisenberg description. The result is a quantum generalization of the \LevKhi\  formula; Gaussian and jump contributions are included. As a byproduct, the most general quasi-free quantum-dynamical semigroup is obtained; on the classical side the Liouville equation and the Kolmogorov-Fokker-Planck equation are included. As a classical subsystem can be observed, in principle, without perturbing it, information can be extracted from the quantum system, even in continuous time; indeed, the whole construction is related to the theory of quantum measurements in continuous time. While the dynamics is formulated to give the hybrid state at a generic time $t$, we show how to extract multi-time probabilities and how to connect them to the quantum notions of \emph{positive operator valued measure} and \emph{instrument}. The structure of the generator of the dynamical semigroup is analyzed, in order to understand how to go on to non quasi-free cases and to understand the possible classical-quantum interactions; in particular, all the interaction terms which allow to extract information from the quantum system necessarily vanish if no dissipation is present in the dynamics of the quantum component. A concrete example is given, showing how a classical component can input noise into a quantum one and how the classical system can extract information on the behaviour of the quantum one.
\end{abstract}

\tableofcontents

\section{Introduction}\label{sec:intro}

Quantum-classical hybrid systems have been studied for various reasons, ranging from computational advantages, description of mesoscopic systems\ldots, to foundational problems, such as the formulation of quantum measurements, see \cite{DamWer22,DGS00,Dio11,Dio23,ManRT23,Sergi+21,Opp+23} and references there in. Quantum measurements can be interpreted as involving hybrid systems: a \emph{positive operator valued measure} is a channel from a quantum system to a classical one, an \emph{instrument} \cite[Sec.\ 4.1.1]{Hol01}  is a channel from a quantum system to a hybrid system \cite{BarL06banach,Dio11,DamWer22,ManRT23}. Also the quantum measurements in continuous time can be interpreted in terms of hybrid systems; here the classical component is the monitored signal extracted from the quantum system \cite{BarL91,BarHL93,BarP96,Hol01,BarPZ98}. A presentation of the theory of measurements in continuous time is developed in Maassen's contribution \cite{Maa23}.

Our aim is to study a class of possible hybrid dynamics in the ``Markov'' case (no memory); so, we need to generalize quantum dynamical semigroups on the quantum side  \cite{Hol01,WisM10}, and semigroups of transition probabilities on the classical side  \cite{Sato99,App09}.  With respect to the theory of quantum measurements in continuous time, we are generalizing the notions of ``convolution semigroups of instruments'' and ``semigroups of probability operators'' \cite{Bar87,BarL89,Bar89,Bar90,BarL90,BarL91,BarHL93,BarP96,BarPZ98,Hol86,Hol89,Bar93}.

As well known, hard mathematical problems arise in the study of dynamical semigroups, when unbounded generators are involved; see for instance \cite{SieHW17,Hol01,ChebF98} for the quantum case. To have a significant class of semigroups, but avoiding these problems, we consider only \emph{quasi-free} transformations, introduced and developed in \cite{DamWer22}. Such transformations are defined by their action on Weyl operators \cite{Van78,DVV79,DamWer22,Hell10}; they generalize the Gaussian case \cite{Hol01}. Quasi-free processes can be characterized by a translation covariance property involving an additional linear map between the phase spaces of input and output systems. The translation covariance would specialize in the classical case to a process whose {\it increments} have a specified distribution. These increments will then be independent by virtue of the Markov property. It is therefore expected that the generators are described by a version of the \LevKhi\  formula. This is indeed a description of our main result. The appearance of quantum versions of the \LevKhi\  formula and of infinitely divisible distributions goes back to the first attempts of constructing a general framework for measurements in continuous time \cite{Hol86,Hol88QPIII,Hol89}.

Our paper is organized as follows.

The hybrid Weyl operators are introduced in Sec.\  \ref{sec:qcW}. The notion of hybrid dynamical semigroup is introduced in Sec.\ \ref{Sec:qfSem}. The main result is the explicit structure of the most general quasi-free hybrid  semigroup. This structure is given in terms of the action on the Weyl operators, and it turns out to to be a quantum generalization of the classical \LevKhi\  formula. The structure of the generator is discussed in Sec.\ \ref{sec:generator}, because it is linked to non quasi-free cases and it helps in understanding the physical interactions.

In Sec.\ \ref{sec:qds} we give the most general quasi-free quantum dynamical semigroup, in the pure quantum case. With respect to the known Gaussian dynamics \cite{HHW10}, also ``jump'' terms are introduced; similar contributions appeared in symmetry based approaches \cite{Hol95,Hol96}. A couple of simple examples are given at the end of the section: a particle in a noisy environment and a quantum harmonic oscillator.

Section \ref{classical} is dedicated to the pure classical case; as a comparison with the quantum case, an example based on the classical harmonic oscillator  is given at the end of the section.
We also show that the deterministic classical Liouville equation is included in the treatment. However, the essential point of this section is to show that not only probability densities at a single time can be constructed, but that the semigroup gives rise also to transition probabilities and multi-time probabilities; then, a whole stochastic process in time is constructed.

The dynamics of a generic hybrid system is studied in Sec.\ \ref{sec:hybrid}.  Now we have an interaction between the quantum and the classical components and a flow of information from the classical sub-system to the quantum one and viceversa. By generalizing the concept of transition probabilities of the classical case, we are able to show that the hybrid dynamics implies the existence of ``transition instruments'' and of multi-time probabilities for the classical component, but which contain information on the quantum system. An example showing the possibility of flow of information in both directions is presented at the end of the section. Conclusions and possible developments are given in Sec.\ \ref{sec:concl}.

\section{Setting and main result}\label{sec:sett+main}

To introduce the notion of quasi-free dynamics we need position and momentum operators for the quantum component and we take the Hilbert space $\Hscr=L^2(\Rbb^n)$; the Lebesgue measure is understood. Then, we take a classical system with $s$ degrees of freedom, living in $\Rbb^s$.

The choice of the spaces of states and of observables can depend on the aims of the construction one wants to realize; a detailed discussion on this point is given in \cite[Secs.\ 2, 3]{DamWer22}. We shall try to define the dynamical semigroup of interest by asking a set of minimal properties and, then, to search for the natural setting. Essentially, there is the possibility of a $W^*$-algebraic approach, which is often well suited to develop a theory of dynamical semigroups for finitely many degrees of freedom, \cite[Secs. 2.8, 4.4, 5.3, 5.7]{DamWer22}, \cite{BarP96}. Alternatively, there are $C^*$-algebraic approaches, needed for instance when classical pure states are important \cite[Sec.\ 2.5]{DamWer22}.

\subsection{Some notations}\label{sec:notations}

Let us start by introducing the various spaces we shall need in the following.

Firstly, we introduce the spaces of operators on $\Hscr=L^2(\Rbb^n)$: 
\begin{itemize}
\item
$\Bscr(\Hscr)$: bounded operators, 
\item $\Tscr(\Hscr)$: trace class, 
\item $\Kscr(\Hscr)$: compact operators. 
\end{itemize}
Then, we introduce the spaces of complex functions on $\Rbb^s$: $L^\infty(\Rbb^s)$, $L^1(\Rbb^s)$, and
\begin{itemize}
\item $C_b(\Rbb^s)$: bounded continuous functions, 
\item $C_0(\Rbb^s)$:  continuous functions vanishing at infinity. 
\end{itemize}

A first choice as observable space is the $W^*$-algebra $\Nscr=\Bscr(\Hscr)\otimes L^\infty(\Rbb^s)$. 
We can identify a generic element $F\in \Nscr$ with a function $F(x)$ from $\Rbb^s$ into $\Bscr(\Hscr)$.
Then, the state space is the predual $\Nscr_*=\Tscr(\Hscr)\otimes L^1(\Rbb^s)$. 
A state $\hat\pi$ is a trace-class valued function $\hat\pi(x)\in \Tscr(\Hscr)$, $x\in \Rbb^s$, such that $\hat\pi(x)\geq 0$ and $\int_{\Rbb^s}\rmd x\, \Tr\{\hat\pi(x)\}=1$. Measurability and integrability properties are always understood.  This choice of state space allows to include generic quantum master equations from one side (Sec.\ \ref{sec:qds}), and, on the classical side, both the Liouville equation and the Kolmogorov-Fokker-Planck equation (Sec.\ \ref{classical}).

Another possible choice, given in \cite[Prop.\ 6]{DamWer22}, is to take $\Kscr(\Hscr)\otimes C_0(\Rbb^s)$ as observable space and its dual as state space.

We denote by $\id$ the identity operator on $\Nscr$ and by $\openone$ the unit element in $\Bscr(\Hscr)$. 
The adjoint of the operator $a\in\Bscr(\Hscr)$ is $a^\dagger$ and $\overline \alpha$ is the complex conjugated of $\alpha\in \Cbb$. For a matrix $M$, not necessarily square, $M^\T$ denote its transpose.

The translation operator on the classical component is denoted by $\Rcal_z$: 
\ $\forall f\in L^\infty(\Rbb^s)$, \ $\Rcal_z[f](x)=f(x+z)$. \  $\Rcal_z$ and $\id\otimes \Rcal_z$ are identified.

We denote by $Q_j$, $P_i$ the position and momentum operators for the quantum component and by  $X_l$ the multiplication operators in the classical component. As in \cite{DamWer22}, we denote by $R$ the vector of the fundamental operators: 
\[
R=\begin{pmatrix}Q\\ P\\ X\end{pmatrix}, \qquad R_j =\begin{cases} Q_j & \text{for } j=1,\ldots, n,
\\ P_{j-n} & j= n+1 ,\ldots, 2n,
\\ X_{j-2n}\quad & j=2n+1,\ldots ,2n+s.
\end{cases}
\]
Then, formally, the commutation rules take the form: 
\begin{equation}\label{sigmaij}
[R_i,R_j]=\rmi\sigma_{ij}\openone, \qquad \sigma_{ij}=\begin{cases}1  & 1\leq i\leq n\quad j=i+n,
\\ -1 \quad & n+1\leq i \leq 2n \quad j=i-n,
\\ 0 &\text{otherwise}. \end{cases}
\end{equation}
\emph{Commutator} and \emph{anti-commutator} are denoted by $[\bullet,\bullet]$ and $\{\bullet,\bullet\}$, respectively.

It is useful to combine the spaces $\Rbb^{2n}$, related to the Hilbert space $\Hscr$, and $\Rbb^s$, appearing in the classical component, in a unique space; so, we define the \emph{phase space} 
\begin{equation*}
\Xi=\Xi_1\oplus \Xi_0, \qquad \Xi_1=\Rbb^{2n}, \qquad \Xi_0=\Rbb^s, \qquad d=2n+s \quad \Rightarrow \quad \Xi=\Rbb^d.
\end{equation*} 
It is also useful to see $\Xi$ and its components as commutative groups under addition (translation groups). Moreover, let $P_{j}$ be the orthogonal projection on $\Xi_j$, $j=1,0$. By using the column notation for vectors, for $\xi, \,\eta\in \Xi$, their scalar product is $\xi^\T \xi=\sum_{i=1}^{d} \xi_i\eta_i$. Finally, let $\sigma $ be the matrix defined by the elements $\sigma_{ij}$ \eqref{sigmaij}.

\subsection{Quantum-classical Weyl operators}\label{sec:qcW}
We introduce now the Weyl operators for the hybrid system \cite[(6)]{DamWer22}.
The parameters of the Weyl operators will live in the space $\Xi$, which is always identified with its dual.
We denote by $\xi$, $\zeta$, $k$, \ldots column vectors and by $\xi^{\T}$, \ldots their transposed versions (row vectors). Often we use
\begin{equation*}
\zeta=\begin{pmatrix} u\\ v\end{pmatrix}\in \Rbb^{2n}\equiv \Xi_1,\qquad k\in \Rbb^s\equiv \Xi_0, \qquad \xi=\begin{pmatrix}\zeta\\ k\end{pmatrix}\in  \Xi.
\end{equation*}

The \emph{Weyl operators} \cite{DamWer22,Hol01} are defined by 
\begin{equation}\label{hybWop}\begin{split}
& W(\xi)=\exp \left\{\rmi R^\T \xi\right\}= W_1(\zeta)W_0(k), \\ &  W_0(k)=\exp \left\{\rmi X^\T k\right\}=\exp \left\{\rmi R^\T P_{0}\xi\right\},
\\
& W_1(u,v)=\exp \left\{\rmi  \left(u^\T Q+v^\T P\right)\right\}=\exp \left\{\rmi R^\T P_1 \xi\right\}.
\end{split}
\end{equation}

\begin{remark}\label{rem:funct}
$W_0(k)$ can be seen as a function from $\Rbb^s$ into $\Cbb$: $W_0(k)(x)=\exp\left\{\rmi x^\T k\right\}$. Moreover, $W_1(\zeta)\in\Bscr(\Hscr)$ and $W_0(k)\in C_b(\Rbb^s)$; therefore, $W(\xi)\in \Nscr$.
\end{remark}

The Weyl operators satisfy the following composition property:
\begin{equation}\label{WxiWxi'}
W(\xi+\eta)
=W(\xi)W(\eta)\exp\left\{\frac\rmi 2\,\xi^\T\sigma\eta\right\}=W(\eta)W(\xi)\exp\left\{-\frac\rmi 2\,\xi^\T\sigma\eta\right\}.
\end{equation}
More rigorously, the Weyl operators $W_1$ are defined as  projective unitary representations of the translation group $\Xi_1$ \cite{Hol01}, or as displacement operators acting on coherent vectors \cite{Parthas92,WisM10}. Then, \eqref{WxiWxi'} represents the rigorous version of the canonical commutation rules \cite{Hol01}.

\subsubsection{Other properties of the Weyl operators} \label{app:Weyl}

By using $W_1$ and the classical translation operator $\Rcal_z$, defined in Sec.\ \ref{sec:notations}, we have the translation properties \cite[(5), (6), (9), Sec.\ 2.6]{DamWer22}:
\begin{subequations}
\begin{equation}\label{classicaltransl}\begin{split}
& \Rcal_z[W_0(k)](x)=W_0(k)(x+z)=\exp\left\{\rmi k^\T(x+z)\right\}=W_0(k)(x)\rme^{\rmi k^\T z}, \\ &\Rcal_z[W(\xi)]=W(\xi)\rme^{\rmi z^\T P_{0}\xi},
\end{split}\end{equation}
\begin{equation}\label{qtranslations}
W_1(u,v)^\dagger Q_iW_1(u,v)=Q_i-v_i,
\qquad W_1(u,v)^\dagger P_iW_1(u,v)=P_i+u_i, \qquad i=1,\ldots,n.
\end{equation}
\end{subequations}
In a shorter form we can write 
\begin{equation}\label{translations}\begin{split}
W(\xi)^\dagger R_iW(\xi)&=R_i+\left(\sigma^\T\xi\right)_i , \qquad i=1,\ldots,2n, \\ &\Rightarrow \qquad \left[\left(\sigma R\right)_i,W(\xi)\right]=\xi_iW(\xi), \qquad i=1,\ldots,2n.
\end{split}\end{equation}

From \eqref{WxiWxi'}, we get the properties:
\begin{subequations}
\begin{equation}\label{W*WW}
W(\eta)^\dagger W(\xi)W(\eta)=W(-\eta)W(\xi+\eta)\rme^{-\frac\rmi 2\,\xi^\T\sigma\eta}=\rme^{\rmi {\eta}^\T\sigma \xi}W(\xi),
\end{equation}
\begin{equation}\label{Weyl:jump}
W(\eta)^\dagger W(\xi)W(\eta)-\frac12\left\{W(\eta)^\dagger W(\eta), W(\xi)\right\}
=W(\eta)^\dagger W(\xi)W(\eta)-W(\xi)=\left(\rme^{\rmi {\eta}^\T\sigma \xi}-1\right)W(\xi).
\end{equation}
\end{subequations}
Also explicit forms of the derivatives of the Weyl operators are very useful in computations:
\begin{equation}\label{WtoR}\begin{split}
\frac{\partial W(\xi)}{\partial \xi_j}&= \frac\rmi 2\left\{ R_j,W(\xi)\right\},\quad j=1,\ldots,d, \\ \rmi k_i W_0(k)(x) &= \frac{\partial \ }{\partial x_i}\, W_0(k)(x), \quad \quad i=1,\ldots,s.
\end{split}\end{equation}

\subsubsection{Characteristic function of a state and Wigner function}\label{sec:Wigf}

As in the pure quantum case, the states $\hat \pi\in\Nscr_*$ are uniquely determined by their characteristic function $\chi_{\hat\pi}(\xi)$ \cite[Sec. 2.4]{DamWer22} or by their Wigner function $\Wcal_{\hat\pi}(z)$ \cite{WisM10}:
\begin{equation}\label{char+Wig}
\chi_{\hat\pi}(\xi)= \int_{\Rbb^s}\rmd x \,\Tr\left\{\hat\pi(x)W(\xi)(x)\right\}, \qquad \Wcal_{\hat \pi}(z)=\frac1{(2\pi)^{d}}\int_{\Xi}\rmd \xi\, \rme^{-\rmi z^\T \xi} \chi_{\hat\pi}(\xi).
\end{equation}
By Bochner's Theorem, the Wigner function is non-negative, $\Wcal_{\hat \pi}(z)\geq 0$, $\forall z\in \Rbb^{d}$, if and only if $\chi_{\hat\pi}(\xi)$ is positive definite: for any choice of the integer $N$, of $\xi_k\in \Xi$, $k=1,\ldots,N$, $c_k\in \Cbb$,
\begin{equation}\label{cposCF}
0\leq \sum_{l,k=1}^N \overline{c_l}\chi_{\hat\pi}(\xi_k-\xi_l)c_k \equiv \sum_{l,k=1}^N\int_{\Rbb^s}\rmd x \, \overline{c_l}\Tr\left\{W(\xi_l)(x)^\dagger\hat\pi(x)W(\xi_k)(x)\right\}c_k\exp\left\{-\frac \rmi 2 \,\xi_k^\T\sigma \xi_l\right\}.
\end{equation}
Due to the last factor, the characteristic function is not positive definite for any state; so, for some states, the Wigner function can become negative for some $z$.

\begin{remark}\label{cf=state} By \cite[Theor.\ 4]{DamWer22}, a function $\chi : \Xi \to \Cbb$ is the characteristic function of a state  if and only if
\\ (1) $\chi$ is continuous, \qquad (2) $\chi(0)=1$,
\qquad
(3) for every integer \ $N$ \ and every choice of \ $\xi_1, \ldots, \xi_N$, \ $\xi_j\in \Xi$, \  the $N\times N$-matrix with elements \  $\chi(\xi_k-\xi_l)\exp\left\{\frac \rmi 2 \,\xi_k^\T\sigma \xi_l\right\}$ \ is positive semi-definite.

To guarantee that the state belongs to $\Nscr_*$ the following property has to be added:
\qquad
(4) $\chi\in L^1(\Xi)$.
\end{remark}

\subsection{A quasi-free dynamical semigroup for a hybrid system}\label{Sec:qfSem}

Here we consider a Markovian dynamics, in the Heisenberg description. With the further condition of translation invariance in the classical component, similar semigroups were introduced inside the theory of quantum measurements in continuous time under the name of \emph{convolution semigroups of instruments}. In that case their generator was characterized under a strong continuity in time, implying that only bounded operators on the quantum component were involved \cite{Hol86,Bar87,Hol89,BarL89,Bar89,BarL90,Bar90,BarL91,BarHL93,BarP96,Hol01,Maa23}. Now we want to treat a generic hybrid system, without the hypothesis of translation invariance and of bounded operators. However, as a first step,  we consider only quasi-free dynamical semigroups,  according to the definition of quasi-free channels given in \cite[Sec.\ 4.1]{DamWer22}.

As motivated in \cite{DamWer22} and recalled at the beginning of Sec.\ \ref{sec:sett+main}, it is better to leave open the choice of the spaces of observables and states. To this end, we give a minimal definition of the semigroup $\{\Tcal_t, \; t\geq 0\}$, involving only its action on the Weyl operators.

\begin{definition}\label{def:new1} A \emph{quasi-free hybrid dynamical semigroup} is a family $\{\Tcal_t, \; t\geq 0\}$ 
of linear operators defined at least on the linear span of the Weyl operators, satisfying the following properties: $\forall t,s\geq 0$,
\begin{itemize}
\item[a.] \label{def:a}(normalization) $\Tcal_t[\openone]=\openone$;
\item[b.] \label{def:b} (initial condition) $\Tcal_0=\id$;
\item[c.] \label{def:c} (semigroup property) $\Tcal_t\circ \Tcal_s=\Tcal_{t+s}$;
\item [d.] \label{def:d}(complete positivity) for any integer $N$ and any choice of the vectors $\phi_k\in \Hscr$ and $\xi^k\in \Xi$, $k=1,\ldots,N$,
\begin{equation}\label{newCPcondition}
\sum_{k,l=1}^N\langle \phi_k|\Tcal_t[W(\xi^k)^\dagger W(\xi^l)]| \phi_l\rangle\geq 0;
\end{equation}
\item[e.] (quasi-free property) \label{def:prop.e} for all $\xi\in \Xi$ 
\begin{equation}\label{Tqf}
\Tcal_t[W(\xi)]=f_t(\xi)W(S_t\xi),
\end{equation}
where $S_t$ 
is a linear operator from $\Xi$ to $\Xi$, and $f_t$ is a continuous function from $ \Xi$ to $\Cbb$;
\item[f.] (continuity in time) the functions $t\to f_t(\xi)$ and  $t\to S_t$ are continuous.
\end{itemize}
\end{definition}

By the property \eqref{WxiWxi'}, the product of Weyl operators is proportional to a Weyl operator, so that the positivity property \eqref{newCPcondition} involves the action of $\Tcal_t$ only on Weyl operators. In \cite{DamWer22} $f_t$ is called \emph{noise function}. By using the composition property \eqref{WxiWxi'} of the Weyl operators and the quasi-free structure \eqref{Tqf}, we get
\begin{equation}\label{TWW}
\Tcal_t[W(\xi^k)^\dagger W(\xi^l)]=f_t(\xi^l-\xi^k)\exp\left\{\frac \rmi 2{\xi^k}^\T \left(\sigma - {S_t}^\T \sigma S_t\right)\xi^l \right\} W(S_t\xi^k)^\dagger W(S_t\xi^l).
\end{equation}

Finding the structure of $f_t(\xi)$ and $S_t$ and constructing explicitly the semigroup requires various steps and involves the theory of classical stochastic  processes, such as processes with independent increments, additive processes, L\'evy processes, and processes of Ornstein-Uhlenbeck type \cite{Sato99,App09,Sko91,SatoY84}. So, here we give only the final results and we leave to Sec.\ \ref{sec:maintheor} the proof of the theorems and the discussion of the construction.

Quasi-free semigroups (on general CCR algebras) were considered also in \cite{Hell10}, and preliminary results on their structure were obtained. Here we arrive to a complete characterization in the case of a finite-dimensional phase space.

The expression of the noise function $f_t(\xi)$ and the proof of Theor.\ \ref{mainTheor} involve infinitely divisible distributions and \LevKhi\  formula. The function $\psi(\xi)$ appearing in the formulation of the theorem is sometimes called \emph{L\'evy symbol} \cite[p.\ 31]{App09}, while $\bigl( A,\, \nu,\, \alpha\bigr)$, appearing in \eqref{formpsi}, is said to be its \emph{generating triplet} \cite[Def.\ 8.2]{Sato99}; $\nu$ is the \emph{L\'evy measure}.
As it is customary for L\'evy measures, the integral \eqref{formpsi} describing the non-Gaussian jump part is separated into small jumps and large jumps using the indicator function $\ind_\Sbb$ of the unit ball $\Sbb=\left\{ \xi\in \Rbb^{d}: \abs\xi<1\right\}$. 
This cutoff function can be changed to any function which is $=1$ for small arguments and vanishes for large ones; this replacement  can be compensated by a change of the vector $\alpha$ (see \cite{Sato99}).

\begin{theorem}\label{mainTheor} The objects $f_t(\xi)$ and $S_t$ describe a  quasi-free hybrid dynamical  semigroup $\{\Tcal_t, \; t\geq 0\}$
in the sense of Definition \ref{def:new1}, if and only if they have the following structure:
\begin{enumerate}
\item \label{prop11} The matrices $S_t$ form a semigroup: there is a real $d\times d$-matrix $Z$ such that 
\begin{equation}\label{Ssemig}
S_t=\rme^{Zt}, \qquad \forall t\geq 0.
\end{equation}

\item \label{mainfprop} The noise function $f_t(\xi)$ is parameterized by a vector $\alpha\in \Xi$, a real symmetric $d\times d$-matrix $A$,
  and a L\'evy measure $\nu$ on $\Xi$, i.e., a $\sigma$-finite measure such that
\begin{equation}\label{propnu}
\nu(\{0\})=0, \qquad \int_{\Sbb} \abs{\eta}^2\nu(\rmd \eta)<+\infty, \qquad \nu(\{\Xi\setminus\Sbb\})<+\infty.
\end{equation}
Then, $f_t$ is determined by
\begin{equation}\label{f+Psi}
f_t(\xi) = \exp\left(\int_0^t\rmd \tau \,\psi(S_\tau \xi)\right),
\end{equation}
where 
\begin{equation}\label{formpsi}
\psi(\xi)=\rmi \alpha^\T \xi -\frac 12 \,\xi^\T A \xi+\int_{\Xi}\nu(\rmd \eta)\left(\rme^{\rmi {\eta}^\T \xi}-1-\rmi \ind_{\Sbb}(\eta) {\eta}^\T \xi\right) , \qquad \forall \xi\in \Xi=\Rbb^{d}.
\end{equation}

\item \label{A+iB} The matrix $A$ satisfies the further positivity condition 
\begin{equation}\label{posmatr}
A \pm \rmi B \geq 0,\qquad B:=\frac 1 2 \left(\sigma Z-Z^\T\sigma^\T\right).
\end{equation}
\end{enumerate}

For every $t\geq 0$, the operator $\Tcal_t$ extends from the linear span of the Weyl operators to $\Nscr$, and this extension is bounded. Moreover, this extension, always denoted by $\Tcal_t$, is the adjoint of a completely positive, normalization preserving map ${\Tcal_t}_\ast$ on the predual $\Nscr_\ast$ such that  $\{{\Tcal_t}_\ast, \ t\geq 0\}$ forms a strongly continuous semigroup.
\end{theorem}

Perhaps the most remarkable feature of this result is that the positivity condition \eqref{posmatr}, which captures all quantum uncertainty constraints of such evolutions,  involves only the Gaussian part, i.e., the matrix $A$, and not the jump part. Since this is also the only place where the symplectic form $\sigma$ enters, the conditions as stated imply those for $\sigma=0$, i.e., for a purely classical system. Understanding this process essentially carries over to the quantum and hybrid cases: whenever the initial state has positive Wigner function, this will be true through the entire evolution, and the Wigner functions follow just the classical process with the same $f_t$ and $S_t$. By linearity this will even be true for non-positive Wigner function, although in this case the probabilistic interpretation is lost. But it is worth noting that the evolution drives towards states with positive Wigner function.

Apart from the above Theorem, the general results given below and in Sec.\ \ref{sec:generator} include the form of the generator of the semigroup, giving the dynamical equations in a way which connects a Lindblad-like form with a Fokker-Planck structure, and a criterion for the existence of an equilibrium state, balancing the effects of a semigroup $S_t$ contracting to the origin and of the noise driving the process away.

\subsection{Proof of the main results}\label{sec:maintheor}
The proof of Theorem~\ref{mainTheor} will be in the following main steps.
\begin{description}
  \item Step 1. \ We evaluate the composition laws for the $f_t$ and $S_t$. If we could assume $f_t(\xi)$ to be differentiable in $t$, this would already imply the form \eqref{f+Psi}, so this is heuristically clear. However, we are {\it not} assuming this; differentiability in time will be a final result.
  \item Step 2. \ We show that the associated classical process ($\sigma=0$) is well-defined. This is in contrast to the positivity conditions for states, where the twisted definiteness does not imply the positivity of a classical object with the same characteristic function (the Wigner function). The difference is made by the divisibility of the process.
  \item Step 3. \ We use the theory of classical additive processes \cite{Sato99} to get the \LevKhi\ structure \eqref{f+Psi}, \eqref{formpsi}. This would be just the standard result for independent increment processes, if we had $S_t=\idty$ for all $t$. The variant we need for general $S_t$ is sometimes called a generalized Ornstein-Uhlenbeck process \cite{SatoY84}.

  \item Step 4. \ We go back to quantum processes to evaluate the complete positivity condition. Heuristically this is the question of how much noise, as described by the Gaussian part $A$ and the jump measure $\nu$, is needed to meet the quantum uncertainty requirement, a conditional complete positivity property for the functional $\psi$. For sufficiency of \eqref{posmatr} we consider the Gaussian case,  and observe that the jumps add a classically positive definite part. This leaves open the possibility of positivity partly ensured by jump noise, with maybe even a vanishing Gaussian part. However, this is ruled out by a scaling argument.
\end{description}

\subsubsection{Step 1: Composition properties}

Let us translate the basic semigroup properties to $f_t$ and $S_t$, using that the map $\Tcal_t$ uniquely determines these objects and is determined by them.

\begin{proposition} \label{lemma:main} Given a semigroup $\Tcal_t$ according to  Definition~\ref{def:new1}, its defining objects $S_t$ and $f_t(\xi)$,  enjoy the following properties.
\begin{enumerate}
\item \label{prop1} The family of matrices $S_t$ forms a semigroup: $S_t=\exp(Zt)$ for some real matrix $Z$.

\item \label{prop02}  The function $f_t(\xi)$ satisfies the normalization condition $f_t(0)=1$ for all $t\geq0$,
and the initial condition $f_0(\xi)=1$ for all $\xi\in\Xi$.
\item \label{prop2} It satisfies the product formula:
\begin{equation}\label{ff}
f_{t+s}(\xi)=f_s(S_t\xi)f_t(\xi), \qquad \forall t,s\geq 0, \quad \forall \xi\in \Xi.
\end{equation}
\end{enumerate}
\end{proposition}

\begin{proof}
By using Eq.\ \eqref{Tqf}, we have that Property a.\ of Def.\ \ref{def:new1} is equivalent to $f_t(0)=1$,  while Property b.\ is equivalent to
$f_0(\xi)=1, \ S_0\xi=\xi$.
By Properties c.\ and e., we get
\[
f_{t+s}(\xi)W(S_{t+s}\xi)=f_t(S_s\xi)f_s(\xi)W(S_tS_s\xi)=f_s(S_t\xi)f_t(\xi)W(S_sS_t\xi) .
\]
If for a certain $\xi$ we have $f_{t+s}(\xi)\neq 0$ we get \eqref{ff}
and $S_{t+s}\xi=S_sS_t\xi$. If  $f_{t+s}(\xi)= 0$ also the product $f_s(S_t\xi)f_t(\xi)$ must vanish and \eqref{ff} holds again. By continuity in $\Xi$ and $f_t(0)=1$, for fixed $\xi$ and $t+s$ it exists $\alpha>0$ such that $f_{t+s}(\alpha\xi)\neq 0$; then, $\alpha S_{t+s}\xi=\alpha S_sS_t\xi$. As a consequence, $S_t$ is a continuous semigroup of real matrices with $S_0=\openone$, so that $S_t=\rme^{Zt}$.
\end{proof}

\begin{remark}\label{rem:iteration} By iteration, the product formula \eqref{ff} gives
\begin{equation}\label{fproduct}
  f_t(\xi)=\prod_{\ell=0}^{m-1} f_{t/m}(S_{\ell t/m}\xi),
\end{equation}
for any integer $m\geq 1$.
Let us also recall that $f_t(\xi)$ is continuous in $t$ and $\xi$.
\end{remark}

\subsubsection{Step 2: Twisted positive definiteness}
The condition on the noise function $f_t$ in \eqref{Tqf}, which ensures the complete positivity \eqref{newCPcondition} of $\Tcal_t$, is a generalization of Bochner's criterion for the positivity of Fourier transforms. Following \cite{DamWer22}, we will say that $f$ is {\it twisted  positive definite with respect to the antisymmetric form $\sigma$}, if the matrix
\begin{equation*}
  M_{jk}=f(\xi_j-\xi_k)\exp\left\{-\frac\rmi2\,\xi_j^\T\sigma\xi_k\right\}
\end{equation*}
is positive semidefinite for any $\xi_1,...,\xi_N\in\Xi$. Since $M\geq0$ implies $M^\dagger=M$, twisted definiteness implies that $f$ is hermitian in the sense that $f(-\xi)=\overline{f(\xi)}$. Moreover, by taking $\xi_k\mapsto-\xi_k$ we  see that twisted definiteness with respect to $\sigma$ implies the same for $-\sigma$. We {\it cannot} conclude that $f$ is also positive definite in the usual sense, i.e., twisted definite with respect to $\sigma=0$. Since for a state the function $f$ is the Fourier transform of the Wigner function, this is the same as saying that there are quantum states with non-positive Wigner function.

For a quasi-free channel $\Tcal(W(\xi))=f(\xi)W(S\xi)$ to be completely positive, it is necessary and sufficient \cite{DamWer22} that $f$ is definite for the difference of input and output symplectic form, i.e., for the antisymmetric matrix
  $\Delta\sigma = \sigma- S^\T\sigma S$.
Again this holds for both signs, and there is no conclusion about zero definiteness. However, we are now looking at the product \eqref{fproduct} and this implies that $f_t$ is positive definite.

We collect these results in the following Proposition (and we give the explicit proofs).

\begin{proposition} \label{+lemma:main} The noise function $f_t$ of Definition \ref{def:new1} enjoys the following properties.
\begin{enumerate}

\item \label{+prop3} The function $f_t$ is \emph{twisted positive definite:} for every integer $N$ and any choice of $\xi^k\in \Xi$ \ and \ $c_k\in \Cbb$, \ $k=1,\ldots, N$, \ we have
\begin{equation}\label{reducedCPcondition}
\sum_{k,l=1}^N \overline{c_k}f_t(\xi^l-\xi^k)\exp\left\{\pm\frac \rmi 2{\xi^k}^\T \left(\sigma - {S_t}^\T \sigma S_t\right)\xi^l \right\} c_l \geq 0.
\end{equation}

\item The function $f_t$ is also positive definite: for every choice of $N$, $\xi$, $c_k$ as above, we have
\begin{equation}\label{fpd}
\sum_{k,l=1}^N \overline{c_k}f_t(\xi^l-\xi^k) c_l \geq 0.
\end{equation}

\item \label{+prop4} For every fixed time $t\geq 0$, the complex function $\xi\to f_t(\xi)$ is the characteristic function of a probability measure on $\Rbb^{d}$.  Also $\xi \to f_s(S_t\xi)$ is the characteristic function of a probability measure on $\Rbb^d$.

\end{enumerate}
\end{proposition}

\begin{proof}
By using \eqref{TWW} and the vectors $\phi_l=W(S_t\xi^l)^\dagger  c_l\psi$ in the positivity condition d.\   of Definition \ref{def:new1}, we get \eqref{reducedCPcondition} with the plus sign in the exponent. By using $N=2$, $\xi^1=\xi$, $\xi^2=-\xi$, $c_1=1$, $c_2= 1$ or $=\rmi$, we get $\overline{f_t(\xi)}=f_t(-\xi)$. In \eqref{reducedCPcondition} with the plus sign we make the replacements $c_l\to \overline{c_l}$, $\xi_l\to -\xi_l$ and we take the complex conjugated of the full expression; by using $\overline{f_t(\xi)}=f_t(-\xi)$, we get \eqref{reducedCPcondition} with the minus sign in the exponent. So, property \ref{+prop3} is proved.

By adding together the two positive expressions appearing in \eqref{reducedCPcondition}, we get
\[
\sum_{k,l=1}^N \overline{c_k}f_t(\xi^l-\xi^k) \cos\left(\frac 1 2\,{\xi^k}^\T \left(\sigma - {S_t}^\T \sigma S_t\right)\xi^l \right) c_l \geq 0.
\]
This gives
\[
\sum_{k,j=1}^N \overline{c_k}f_{t/m}\big(S_{(l-1)t/m}(\xi^j-\xi^k)\big) \cos\left(\frac 1 2\big(S_{(l-1)t/m}\xi^k\big)^\T \left(\sigma - {S_{t/m}}^\T \sigma S_{t/m}\right)S_{(l-1)t/m}\xi^j \right) c_j \geq 0.
\]
As the element-wise (Schur-Hadamard) product of non-negative matrices is non-negative, we have
\[
\sum_{k,j=1}^N \overline{c_k} \, c_j \prod_{l=1}^m\biggl[f_{t/m}\big(S_{(l-1)t/m}(\xi^j-\xi^k)\big) \cos\left(\frac 1 2\big(S_{(l-1)t/m}\xi^k\big)^\T \left(\sigma - {S_{t/m}}^\T \sigma S_{t/m}\right)S_{(l-1)t/m}\xi^j \right)\biggr] \geq 0;
\]
by  \eqref{fproduct}, this gives
\[
\sum_{k,j=1}^N \overline{c_k}f_{t}(\xi^j-\xi^k)c_j \prod_{l=1}^m\left[\cos\left(\frac 1 2\big(S_{(l-1)t/m}\xi^k\big)^\T \left(\sigma - {S_{t/m}}^\T \sigma S_{t/m}\right)S_{(l-1)t/m}\xi^j \right)\right]  \geq 0.
\]
Being $S_t$ a matrix semigroup with $S_0=\openone$, clearly, for large $m$, the argument of each cosine is $\Order(t/m)$, so each cosine is $1+\Order\big((t/m)^2\big)$. Hence the whole product goes to $1$ as $m\to\infty$:
\[
\lim_{m\to+\infty}\prod_{l=1}^m\left[\cos\left(\frac 1 2\big(S_{(l-1)t/m}\xi^k\big)^\T \left(\sigma - {S_{t/m}}^\T \sigma S_{t/m}\right)S_{(l-1)t/m}\xi^j \right)\right] =1.
\]
Then, we obtain \eqref{fpd}.

By the normalization $f_t(0)=1$, the positive definite condition \eqref{fpd}, and the continuity of $\xi\to f_t(\xi)$, we have that $f_t(\cdot)$ is the characteristic function of a probability measure on $\Xi$ (Bochner's theorem, see \cite[Prop.\ 2.5, point (i)]{Sato99}). The same arguments give that $f_s(S_t \cdot) $ is the characteristic function of some probability measure. So, also property \ref{+prop4} is proved.
\end{proof}

\begin{remark}\label{rem:chfunct} Being $f_t$ a characteristic function, we have also $\overline{f_t(\xi)}=f_t(-\xi)$; then, equation \eqref{Tqf} implies $\Tcal_t[W(\xi)]^\dagger=\Tcal_t[W(\xi)^\dagger]$.
\end{remark}

\subsubsection{Step 3: \LevKhi\  analysis}\label{sec:step3}

In this section we completely ignore the quantum aspects of the problem; we prove the special case of the theorem for purely classical systems ($\sigma=0$). We take into account the necessary positivity condition \eqref{fpd}, but we postpone to the next step the introduction of the more restrictive condition \eqref{reducedCPcondition}.
Now, $\left\{f_t\right\}_{t\geq 0}$ is a collection of classical characteristic functions satisfying the product formula \eqref{ff},
and the problem of finding its structure fits exactly into the framework of additive processes in the sense of Sato \cite[Sec.\ 9]{Sato99}.

An \emph{additive process} $\tilde X_t$ is defined in \cite[p.3]{Sato99} to be a stochastic process with independent increments, starting from the origin $\tilde X_0=0$, and enjoying stochastic continuity. In \cite[Chapt.\ 3]{Sko91} these processes are simply called  ``stochastically continuous processes with independent increments''. In our context the ``independence of the increments'' essentially follows from the product properties \eqref{ff}, \eqref{fproduct}, as we shall see in the following.

Let us set
\begin{equation}\label{hatf}
\hat f_{s,t}(\xi):= f_{t-s}(S_s \xi), \qquad 0\leq s\leq t.
\end{equation}
By property \ref{+prop4} of Proposition \ref{+lemma:main}, $\hat f_{s,t}(\cdot)$ is the characteristic function of a probability measure on $\Rbb^d$ which we denote by $\mu_{s,t}$. By the product formula \eqref{ff} written with $\xi$ replaced by $S_{t_1}\xi$, $t_1\geq 0$, and by setting $t_2=t_1+t$,  $t_3=t_2+s$, we get
\begin{equation*}
\hat f_{t_1,t_2}(\xi)\hat f_{t_2,t_3}(\xi)=\hat f_{t_1,t_3}(\xi), \qquad 0\leq t_1\leq t_2\leq t_3.
\end{equation*}
This means that for the associated probability measures we have
\begin{equation}\label{muconvol}
\mu_{t_1,t_2}*\mu_{t_2,t_3}=\mu_{t_1,t_3};
\end{equation}
the symbol $*$ denotes the convolution. By applying the results of \cite[Sec.\ 9]{Sato99}, we can find the main properties of the process associated to these probability measures and the explicit structure of our noise function $f_t(\xi)$.

\begin{proposition} \label{prop:LK} The following statements hold:
\begin{enumerate}
\item \label{addproc!} There exists an additive process in law $\tilde X_t$ on $\Rbb^d$ such that $\mu_{s,t}$ is the probability measure of the increment $\tilde X_t-\tilde X_s$, \ $t\geq s\geq 0$.

\item \label{muinfdiv} The probability measure $\mu_{s,t}$ is infinitely divisible.

\item \label{Psiexistence} There is a unique continuous function $\Psi_t(\cdot)$ from $\Rbb^d$ into $\Cbb$, such that
\begin{equation}\label{f=expPsi}
f_t(\xi)=\rme^{\Psi_t(\xi)}, \qquad \forall \xi \in \Rbb^d, \quad \forall t\geq 0;
\end{equation}
$\Psi_t(\xi)$ is also continuous in time and
\begin{equation*}
\Psi_0(\xi)=0, \qquad \Psi_t(0)=0.
\end{equation*}

\item \label{ttriplet} For all $t\geq 0$, the function $\Psi_t(\xi)$ admits the representation: \ $ \forall \xi\in \Xi=\Rbb^{d}$,
\begin{equation}\label{formPsit}
\Psi_t(\xi)=\rmi \alpha(t)^\T \xi -\frac 12 \,\xi^\T A_t \xi+\int_{\Rbb^d}\nu_t(\rmd \eta)\left(\rme^{\rmi {\eta}^\T \xi}-1-\rmi \ind_{\Sbb}(\eta) {\eta}^\T \xi\right) ,
\end{equation}
where $\alpha(t)\in \Rbb^d$, $A_t$ is a symmetric nonnegative definite $d\times d$ matrix, and $\nu_t$ is a measure on $\Rbb^d$ satisfying
\[
\nu_t(\{0\})=0 \qquad \text{and} \qquad \int_{\Rbb^d}\left(\abs \xi^2 \wedge 1\right)\nu_t(\rmd \xi) < +\infty.
\]
The representation of $\Psi_t(\xi)$ by the triplet $\left( A_t,\, \nu_t,\, \alpha(t)\right)$ is unique.

\item \label{triplprop} The triplet $\left( A_t,\, \nu_t,\, \alpha(t)\right)$ enjoys the following properties:
\begin{enumerate}
\item
$
A_0=0, \qquad \nu_0=0,\qquad \alpha(0)=0$;
\item
if $0\leq s\leq t<+\infty$, then $A_s\leq A_t$ and $\nu_s(E)\leq \nu_t(E)$
for every Borel  subset $E$ of $\Rbb^d$;
\item
as $s\to t$ in $(0,+\infty)$, we have $A_s\to A_t$,  \ $\alpha(s)\to \alpha(t)$, and $\nu_s(E)\to \nu_t(E)$ for every Borel  subset $E$ of $\Rbb^d$ such that $E\subset \{x:\abs x >\epsilon\}$, $\epsilon>0$.
\end{enumerate}
\end{enumerate}
\end{proposition}

\begin{proof}
Let us consider Properties (9.13)-(9.16) in Theorem 9.7 of \cite{Sato99}. Condition (9.13) is exactly \eqref{muconvol}. By the definition \eqref{hatf}, we have $\hat f_{s,s}(\xi)= f_0(\xi)=1$, so that $\mu_{s,s}=\delta_0$, which is condition (9.14).
By the continuity in time of $S_t$ and $f_t(\xi)$ and the definition \eqref{hatf}, we have the continuity in $s$ and $t$ of $\hat f_{s,t}(\xi)$. The weak convergence of probability measures is defined in \cite[Def.\ 2.2]{Sato99} and it is equivalent to the convergence of the characteristic functions, see Properties (vi) and (vii) of Proposition 2.5. These facts easily imply that also conditions (9.15) and (9.16) hold.
Then, by point (ii) of \cite[Theor.\ 9.7]{Sato99}, point \ref{addproc!} is proved.

By Theor.\ 9.1 in Ref.\ \cite{Sato99}, the probability measure $\mu_{s,t}$ is infinitely divisible, which is point \ref{muinfdiv}.

$f_t$ is the characteristic function of the infinitely divisible distribution $\mu_{0,t}$, which gives $f_t(\xi)\neq 0$, $\forall \xi\in \Rbb^d$. Let us define $f_{-t}(\xi)=1$ for $t>0$. By the continuity in time and on the phase space, the complex function $(t,\xi)\to f_t(\xi)$ on $\Rbb^{d+1}$ is continuous; moreover,  $f_t(0)=1$, \ $f_t(\xi)\neq 0$, $\forall \xi\in \Rbb^d$. Then, the statement \ref{Psiexistence} follows from \cite[Lemma 7.6]{Sato99}.

By points \ref{muinfdiv} and \ref{Psiexistence}, $\Psi_t(\xi)$ is the L\'evy symbol associated to the infinitely divisible distribution $\mu_{0,t}$. By  \cite[Theor.\ 8.1 and Def.\ 8.2]{Sato99}, $\Psi_t(\xi)$ has the structure given in point \ref{ttriplet} and this representation is unique.

By the existence of the additive process of point \ref{addproc!}, we have that the hypotheses of point (i) in \cite[Theor.\ 9.8]{Sato99} hold. Then, we get Properties (1)-(3) of that theorem, which correspond to Properties (a)-(c) in point \ref{triplprop}.
\end{proof}

\begin{remark}\label{+intpsi} By Proposition \ref{prop:LK}, $f_t(\xi) = \rme^{\Psi_t(\xi)}$ is the characteristic function of an infinitely divisible distribution. The structure \eqref{formPsit} is due to the \LevKhi\ formula; $\big( A_t\,, \nu_t\,, \alpha(t)\big)$ is the \emph{generating triplet} of $\Psi_{t}(\xi)$.  This representation is unique once the cutoff function $\ind_\Sbb$ has been fixed. Other choices of the cutoff function are possible; a change of cutoff implies a consequent change of the vector $\alpha(t)$ \cite[Remark 8.4]{Sato99}.
\end{remark}

\begin{remark}\label{rem:fff}
By the result \eqref{f=expPsi}, the product formula \eqref{ff} is equivalent to
\begin{equation}\label{Psit+s}
 \Psi_{t+s}(\xi)=\Psi_t(\xi)+\Psi_s(S_t\xi),\qquad  \forall t,s>0, \quad \forall \xi \in \Xi.
\end{equation}
We have also
\begin{equation*}
\hat f_{s,t}(\xi)=\exp\left\{\Psi_{s,t}(\xi)\right\}, \qquad \Psi_{s,t}(\xi)=\Psi_{t-s}(S_s\xi), \qquad 0\leq s \leq t.
\end{equation*}
\end{remark}

\begin{proposition}\label{prop:tripl=int} The triplet $\left( A_t,\, \nu_t,\, \alpha(t)\right)$, introduced in point \ref{ttriplet} of Proposition \ref{prop:LK}, has the structure
\begin{equation}\label{genetrpl}
A_t=\int_0^t\rmd \tau \,S_\tau^\T A S_\tau, \qquad \nu_t(E)= \int_{\Rbb^d} \nu(\rmd \eta)\int_0^t\rmd \tau \,\ind_E\left( S_\tau^\T \eta\right),
\end{equation}
$E$ is any Borel subset of $\Rbb^d$,
\begin{equation}\label{alpha(t)}
\alpha(t)=\int_0^t\rmd \tau \, S_\tau^\T\alpha +\int_{\Rbb^d} \nu(\rmd \eta)\int_0^t\rmd \tau\, S_\tau^\T\eta \left(\ind_\Sbb\big( S_\tau^\T\xi'\big)-\ind_\Sbb(\eta)\right) .
\end{equation}
Here, $\alpha\in \Rbb^d$, \ $A$ is a real symmetric $d\times d$-matrix with $A\geq 0$,
and $\nu$ is a L\'evy measure  on $\Rbb^d$, that is, it satisfies
\eqref{propnu}.
\end{proposition}

\begin{proof}
By iterating the decomposition \eqref{Psit+s} we obtain
\[
\Psi_t(\xi)=\sum_{l=0}^{m-1}\Psi_{t/m}\left( S_{lt/m}\xi\right), \qquad \forall m\geq 1, \quad t\geq 0, \quad \xi\in \Rbb^d.
\]
By \eqref{formPsit} and the uniqueness of the Gaussian contribution in the generating triplet, we have
\begin{equation}\label{Atdecomp}
\xi^\T A_t \xi = \xi^\T\sum_{l=0}^{m-1}
\left(S_{lt/m}\xi\right)^\T A_{t/m} S_{lt/m}\xi,
\end{equation}
\begin{multline}\label{alphanudecomp}
\rmi \alpha(t)^\T \xi+\int_{\Rbb^d}\nu_t(\rmd \eta)\left(\rme^{\rmi {\eta}^\T \xi}-1-\rmi \ind_{\Sbb}(\eta) {\eta}^\T \xi\right) \\ {}= \sum_{l=0}^{m-1}\left\{\rmi \alpha(t/m)^\T S_{lt/m} \xi +\int_{\Rbb^d}\nu_{t/m}(\rmd \eta)\left(\rme^{\rmi {\eta}^\T S_{lt/m}\xi}-1-\rmi \ind_{\Sbb}(\eta) {\eta}^\T S_{lt/m}\xi\right)\right\}.
\end{multline}
Moreover, by point \ref{triplprop}  in Proposition \ref{prop:LK}, we have
\begin{equation}\label{to0}
\lim_{m\to+\infty}A_{t/m}=0, \qquad \lim_{m\to+\infty}\alpha_{t/m}=0, \qquad \lim_{m\to+\infty}\nu_{t/m}(E)=0
\end{equation}
for every Borel  subset $E$ of $\Rbb^d$ such that $E\subset \{x:\abs x >\epsilon\}$, $\epsilon>0$.

From \eqref{Atdecomp} we get
\begin{multline*}
\xi^\T A_t \xi -  \int_0^t\rmd \tau \,\xi^\T S_\tau^\T \,\frac m t A_{t/m} S_\tau\xi
= \xi^\T\sum_{l=0}^{m-1}\left(
S_{lt/m}^\T A_{t/m} S_{lt/m}-\int_{lt/m}^{(l+1)t/m}\rmd \tau S_\tau^\T\frac { m }t\, A_{t/m}S_\tau\right)\xi
\\ {}= \xi^\T\sum_{l=0}^{m-1}S_{lt/m}^\T\left( A_{t/m} -
\int_{0}^{t/m}\rmd \tau S_\tau^\T\frac { m }t\, A_{t/m}S_\tau \right)S_{lt/m}\xi,
\end{multline*}
\begin{multline*}
\lim_{m\to +\infty} \left(\xi^\T A_t \xi -  \int_0^t\rmd \tau \,\xi^\T S_\tau^\T \,\frac m t A_{t/m} S_\tau\xi\right)
=\lim_{m\to +\infty} \xi^\T\sum_{l=0}^{m-1}S_{lt/m}^\T\left( A_{t/m} -
\int_{0}^{t/m}\rmd \tau S_\tau^\T\frac { m }t\, A_{t/m}S_\tau \right)S_{lt/m}\xi
\\ {}= -\lim_{m\to +\infty}\xi^\T\sum_{l=0}^{m-1}\frac t{2m} S_{lt/m}^\T\left(
A_{t/m}Z+Z^\T  A_{t/m} \right)S_{lt/m}\xi
\\ {}= -\frac 12\,\lim_{m\to +\infty}\xi^\T\int_0^t\rmd \tau\left( S_{\tau}^\T
A_{t/m}ZS_\tau+S_{\tau}^\T Z^\T  A_{t/m} S_\tau\right)\xi=0;
\end{multline*}
the last equality is due to \eqref{to0}. The map $K\to  \int_0^t\rmd \tau \, S_\tau^\T K S_\tau$, at least for small times, is invertible, because it is near to the identity map and positivity preserving; then,
we have the existence of the limit
\[
\lim_{m\to+\infty}\frac { m }t\, A_{t/m}= A \geq 0;
\]
then, the first equality in \eqref{genetrpl} holds.

By \eqref{alphanudecomp} we have
\begin{multline*}
\rmi \alpha(t)^\T \xi -\rmi\, \frac { m }t\, \alpha(t/m)^\T\int_0^t\rmd \tau \, S_\tau\xi
+\int_{\Rbb^d}\nu_t(\rmd \eta)\left(\rme^{\rmi {\eta}^\T \xi}-1-\rmi \ind_{\Sbb}(\eta) {\eta}^\T \xi\right)
\\ {} - \frac m t\int_{\Rbb^d}\nu_{t/m}(\rmd \eta)\int_0^t\rmd \tau\left(\rme^{\rmi {\eta}^\T S_{\tau}\xi}-1-\rmi \ind_{\Sbb}(\eta) {\eta}^\T S_{\tau}\xi\right)
\\ {}= \int_{\Rbb^d}\nu_{t/m}(\rmd \eta)\sum_{l=0}^{m -1}\biggl(\rme^{\rmi {\eta}^\T S_{lt/m}\xi}-1-\rmi \ind_{\Sbb}(\eta) {\eta}^\T S_{lt/m}\xi - \frac m t \int_0^{t/m}\rmd \tau\left(\rme^{\rmi {\eta}^\T S_{lt/m}S_{\tau}\xi}-1-\rmi \ind_{\Sbb}(\eta) {\eta}^\T S_{lt/m}S_{\tau}\xi\right)\biggr)\\ {}+\rmi \alpha(t/m)^\T \left(\sum_{l=0}^{m-1}\left(S_{lt/m}-\frac m t\int_{t_l}^{t_l+t/m}\rmd \tau \, S_\tau\right)\right)\xi
\\ {}= \int_{\Rbb^d}\nu_{t/m}(\rmd \eta)\sum_{l=0}^{m -1}\biggl[\rme^{\rmi {\eta}^\T S_{lt/m}\xi}\biggl(1 - \frac m t \int_0^{t/m}\rmd \tau\rme^{\rmi {\eta}^\T S_{lt/m}\left(S_{\tau}-\openone\right)\xi}\biggr) \\ {}-\rmi \ind_{\Sbb}(\eta) {\eta}^\T S_{lt/m}\left(1-\frac m t \int_0^{t/m}\rmd \tau S_{\tau}\right)\xi \biggr]
+ \rmi\alpha(t/m)^\T\left(\sum_{l=0}^{m-1}S_{lt/m}\left(1-\frac m t\int_{0}^{t/m}\rmd \tau \, S_\tau\right)\right)\xi.
\end{multline*}
Then, for $m\to +\infty$, we get
\begin{multline*}
\lim_{m\to + \infty}\biggl\{\rmi \alpha(t)^\T \xi
-\rmi\, \frac { m }t\, \alpha(t/m)^\T\int_0^t\rmd \tau \, S_\tau\xi +\int_{\Rbb^d}\nu_t(\rmd \eta)\left(\rme^{\rmi {\eta}^\T \xi}-1-\rmi \ind_{\Sbb}(\eta) {\eta}^\T \xi\right) \\ {} - \frac m t\int_{\Rbb^d}\nu_{t/m}(\rmd \eta)\int_0^t\rmd \tau\left(\rme^{\rmi {\eta}^\T S_{\tau}\xi}-1-\rmi \ind_{\Sbb}(\eta) {\eta}^\T S_{\tau}\xi\right)\biggr\}
\\{}
=- \lim_{m\to +\infty}\sum_{l=0}^{m -1}\biggl\{\int_{\Rbb^d}\nu_{t/m}(\rmd \eta)\,\frac{\rmi t} {2m}\left(\rme^{\rmi {\eta}^\T S_{lt/m}\xi}-\ind_{\Sbb}(\eta) \right) {\eta}^\T S_{lt/m}Z\xi
+ \frac{\rmi t} {2m}\,\alpha(t/m)^\T S_{lt/m}Z\xi\biggr\}
\\{}
=- \frac \rmi 2\lim_{m\to +\infty}\int_0^t\rmd \tau \biggl\{\int_{\Rbb^d}\nu_{t/m}(\rmd \eta)\left(\rme^{\rmi {\eta}^\T S_{\tau}\xi}-\ind_{\Sbb}(\eta) \right) {\eta}^\T \dot S_{\tau}\xi+ \alpha(t/m)^\T \dot S_{\tau}\xi\biggr\}
\\{}
= \frac 1 2\lim_{m\to +\infty}\biggl\{\int_{\Rbb^d}\nu_{t/m}(\rmd \eta)\left(\rme^{\rmi \eta^\T\xi}- \rme^{\rmi {\eta}^\T S_{t}\xi}+\rmi \ind_{\Sbb}(\eta)  {\eta}^\T \left( S_{t}-\openone\right)\xi\right)
-\rmi \alpha(t/m)^\T \left( S_{t}-\openone\right)\xi\biggr\}=0;
\end{multline*}
the last equality is due to \eqref{to0}. Then, we must have
\[
\nu_t(E) =\lim_{m\to +\infty}\frac m t \int_0^t\rmd \tau \int_{\Rbb^d}\ind_{E}(S_\tau^\T\eta) \nu_{t/m}(\rmd \eta), \qquad  \alpha(t)^\T \xi
=\lim_{m\to +\infty} \frac { m }t\, \alpha(t/m)^\T\int_0^t\rmd \tau \, S_\tau\xi;
\]
these equations give
\[
\lim_{\epsilon\downarrow 0}\frac{\nu_\epsilon (E)}\epsilon=\nu(E),\qquad \lim_{\epsilon\downarrow 0}\frac { \alpha(\epsilon)}\epsilon= \alpha\in \Xi,
\]
the measure $\nu$ satisfies \eqref{propnu}.

Then, one easily checks that 
\eqref{genetrpl} and \eqref{alpha(t)} hold.
\end{proof}

By Eqs.\ \eqref{f=expPsi}, \eqref{formPsit} and Proposition \ref{prop:tripl=int}, we get the structure of Eqs.\ \eqref{propnu}, \eqref{f+Psi}, \eqref{formpsi} in Theor.\ \ref{mainTheor}:
\begin{equation}\label{ftsummary}
f_t(\xi)=\rme^{\Psi_t(\xi)}, \qquad \Psi_t(\xi)=\int_0^t\rmd \tau \,\psi(S_\tau \xi),
\end{equation}
where the function $\psi$ is given in \eqref{formpsi}.

\subsubsection{Step 4: Complete positivity and sufficiency}

Up to now we have proved that conditions \ref{prop11} and \ref{mainfprop} (with $A\geq 0$) in Theorem \ref{mainTheor} are necessary. Now we have to prove condition \ref{A+iB} (due to complete positivity) and the sufficiency of the three conditions. We have also to prove the final statements in Theorem \ref{mainTheor} about the extension of $\Tcal_t$. Here we give this final part of the proof of the main theorem.

\begin{proof} [Proof of Theorem \ref{mainTheor}] Let us consider that $f_t(\xi)$ is twisted positive definite (point \ref{+prop3}  in Proposition \ref{+lemma:main}). Under the condition $\sum_{k=1}^N c_k=0$, the time derivative of \eqref{reducedCPcondition} in $t=0$ is not negative. By using \eqref{Ssemig} and \eqref{ftsummary}, we get
\begin{equation*}
\sum_{k,l=1}^N\overline{c_k}\left(\psi(\xi^k-\xi^l) \pm\frac \rmi 2 \,{\xi^k}^\T\left(\sigma Z+Z^\T\sigma\right)\xi^l \right)c_l\geq 0.
\end{equation*}
By inserting \eqref{formpsi} into this equation, we obtain
\begin{equation*}
\sum_{k,l=1}^N\overline{c_k}\left(\xi^k A\xi^l \pm \frac \rmi 2 \,{\xi^k}^\T\left(\sigma Z+Z^\T\sigma\right)\xi^l+ \int_{\Xi}\nu(\rmd \eta)\rme^{\rmi {\eta}^\T \left(\xi^k-\xi^l\right)}\right)c_l\geq 0, \qquad \sum_{k=1}^N c_k=0.
\end{equation*}
By dilating  $\xi^k$ into $\lambda\xi^k$, $\lambda\in \Rbb$, and dividing by $\lambda^2$, we get
\begin{equation*}
0\leq \sum_{k,l=1}^N\overline{c_k}\left(\xi^k A\xi^l \pm\frac \rmi 2 \,{\xi^k}^\T\left(\sigma Z+Z^\T\sigma\right)\xi^l+ \frac 1{\lambda^2}\int_{\Xi}\nu(\rmd \eta)\rme^{\rmi \lambda{\eta}^\T \left(\xi^k-\xi^l\right)}\right)c_l,
\end{equation*}
which goes to $\sum_{k,l=1}^N\overline{c_k}\left(\xi^k A\xi^l +\frac \rmi 2 \,{\xi^k}^\T\left(\sigma Z+Z^\T\sigma\right)\xi^l\right)c_l$ for $\lambda \to \pm \infty$; then, \eqref{posmatr} is a necessary condition for positivity.
This ends the proof of the necessity of conditions \ref{prop11}--\ref{A+iB}.

By inserting \eqref{f+Psi} and \eqref{formpsi} into \eqref{Tqf} and by using \eqref{TWW}, it is easy to check that the conditions  \ref{prop11}, \ref{mainfprop}, \ref{A+iB} in Theor.\ \ref{mainTheor} are also sufficient to get all the points in Definition \ref{def:new1}.

By linearity, the operator $\Tcal_t$ is extended to the linear span in $\Nscr$ of the Weyl operators \eqref{hybWop}. By \eqref{newCPcondition}, $\Tcal_t$ sends positive elements into positive ones; if $F\in \Nscr$ is in the linear span of the Weyl operators and $F\geq 0$, we have $F\leq \norm{F}\openone$, and, by the normalization of $\Tcal_t$, we get
\begin{equation}\label{Tbound}
\Tcal_t[F]\leq \norm F \Tcal_t[\openone]= \norm F \openone.
\end{equation}
Let $\hat \pi_0$ be a generic state in $\Nscr_*=\Tscr(\Hscr)\otimes L^1(\Rbb^s)$ with characteristic function $\chi_{\hat\pi_0}(\xi)$, defined in \eqref{char+Wig}. Then, we define
\[
\int_{\Rbb^s}\rmd x \,\Tr\left\{{\Tcal_t}_*[\hat\pi_0](x)W(\xi)(x)\right\}:=\int_{\Rbb^s}\rmd x \,\Tr\left\{\hat\pi_0(x)\Tcal_t[W(\xi)](x)\right\}.
\]
By \eqref{Tqf}, we get
\[
\int_{\Rbb^s}\rmd x \,\Tr\left\{{\Tcal_t}_*[\hat\pi_0](x)W(\xi)(x)\right\}=f_t(\xi)\int_{\Rbb^s}\rmd x \, \rme^{\rmi x^\T P_0S_t\xi} \Tr\left\{\hat\pi_0(x)W_1(P_1S_t\xi)\right\}=: \chi_{\hat \pi_t}(\xi).
\]
By the properties of $f_t(\xi)$ and $S_t$, one can check that all the four properties in Remark \ref{cf=state} hold; then, $\chi_{\hat \pi_t}(\xi)$ is the characteristic function of a state $\hat \pi_t$ in $\Nscr_*$.
Then, we can define ${\Tcal_t}_*$ on $\Nscr_*$ by  ${\Tcal_t}_*[\hat\pi_0]=\hat \pi_t$ and linear extension. One can check that $\{{\Tcal_t}_*, \; t\geq 0\}$ is a strongly continuous semigroup of completely positive, normalization preserving maps.

By defining $\Tcal_t={{\Tcal_t}_*}^*$ on the whole $\Nscr$, by \eqref{Tbound}, we get that $\Tcal_t$ is bounded with norm 1. This is the unique linear extension to $\Nscr$  of the operator of Definition \ref{def:new1}.

This ends the proof of the theorem.
\end{proof}

\begin{remark}\label{rem:SZ}
Being $Z$ the generator of $S_t$, from \eqref{posmatr} we get the following positivity condition on $A_t$, given by the first equality in \eqref{genetrpl}:
\begin{equation}\label{Atpos}
0\leq \int_0^t\rmd \tau \, S_\tau^\T \left[A\pm \frac \rmi 2\left(\sigma Z+Z^\T\sigma\right)\right]S_\tau =A_t\pm \frac\rmi 2 \left(S_t^\T \sigma S_t-\sigma\right).
\end{equation}
\end{remark}

\subsection{Wigner function and equilibrium}\label{sec:Wig}

We can introduce the characteristic function of the hybrid state $\hat\pi_t$ (see Sec.\ \ref{sec:Wigf}) \cite[Theor.\ 4]{DamWer22}; by using the explicit form \eqref{Tqf}, \eqref{f+Psi} of the action of $\Tcal_t$ on the Weyl operators, we get
\begin{equation}\label{Tchi}
\chi_{\hat\pi_t}(\xi)= \int_{\Rbb^s} \rmd y\,\Tr\left\{\hat\pi_t(y)W(\xi)(y)\right\}= \rme^{\Psi_t(\xi)}\chi_{\hat \pi_0}(S_t\xi).
\end{equation}
As one can check, condition \eqref{Atpos} guarantees the positivity condition \eqref{cposCF} for $\chi_{\hat\pi_t}(\xi)$.

The quantity $f_t(\xi)=\rme^{\Psi_t(\xi)}$ is the characteristic function of an infinitely divisible distribution. Its Fourier transform is non-negative because it is a probability density; possibly, the probability has discrete components and the ``density'' has singular components. Formally, we can write
\[
\Wcal_{\Psi_t}(z):=\frac1{(2\pi)^{2n+s}}\int_{\Xi}\rmd \xi\, \exp\left\{-\rmi z^\T \xi+\Psi_t(\xi)\right\}\geq 0.
\]
Then, by \eqref{char+Wig}, we get the Wigner function of the hybrid state $\hat\pi_t$:
\begin{equation}\label{genWig}
\Wcal_{\hat\pi_t}(z)=\frac1{(2\pi)^{2n+s}}\int_{\Xi}\rmd \xi\, \rme^{-\rmi z^\T\tilde \sigma \xi} \rme^{\Psi_t(\xi)} \chi_{\hat \pi_0}(S_t\xi) =\frac1{(2\pi)^{2n+s}}\int_{\Xi}\rmd x\int_{\Xi}\rmd \xi\, \Wcal_{\Psi_t}(x)\chi_{\hat \pi_0}(S_t\xi)\rme^{-\rmi\left( z-x\right)^\T\tilde \sigma \xi} .
\end{equation}
This Wigner function can be also expressed as the convolution
\[
\Wcal_{\hat\pi_t}(z) =\frac1{\abs{\det S_t}}\int_{\Xi}\rmd x\, \Wcal_{\Psi_t}(x)\Wcal_{\hat \pi_0}\big( S_t^\T\left(z-x\right)\big).
\]
The presence of possible negative zones in $\Wcal_{\hat \pi_0}$ depends on the choice of the state $\hat \pi_0$; the convolution with the positive function $\Wcal_{\Psi_t}(x)$ tends to diminish the possible negative contributions. We can say that the quasi-free evolution tends to diminish the quantum signature of the Wigner function.

\subsubsection{Approach to equilibrium}\label{sec:equil}

By using the Wigner function \eqref{genWig} of the hybrid state, it is easy to give a formula for the final equilibrium state of a quasi-free semigroup.
As $f_t(\xi)$ is the characteristic function of a probability distribution, the study of its behaviour for large times is a purely classical problem; results on this problem are given in \cite{SatoY84}.

Let us assume that $S_t$ is a contracting semigroup, that is  $S_t\xi\to0$ as $t\to\infty$. Since it is a matrix semigroup, this condition implies also that $\abs{S_t\xi}\leq e^{-c t}\abs\xi$ for some constant $c>0$. Then, in \eqref{Tchi} the second factor converges to $\chi_{\hat \pi_0}(0)=1$,
i.e. the initial state becomes irrelevant for the limit. For what concerns the first factor,
we assume that the following integral converges:
\begin{equation*}
\int_{\abs \xi>1}\ln\abs \xi \nu(\rmd \xi)<+\infty.
\end{equation*}
By \cite[Theor.\ 4.1]{SatoY84}, under these hypotheses we have the existence of the limit $\Psi_\infty(\xi)=\int_0^{+\infty}\rmd \tau \,\psi(S_\tau \xi)$; moreover, $\rme^{\Psi_\infty(\xi)}$ is the characteristic function of an infinitely divisible distribution.

Under the same hypotheses, from \eqref{genWig}, we have also $\Wcal_{\hat\pi_t}(z) \to \Wcal_{\Psi_\infty}(z)$ which is non-negative, but not necessarily a $L^1$-function: one can have a limit state, but, perhaps,  not  in $\Nscr_*$. A $C^*$- algebraic approach is needed to include the limit state; see the discussion in Sec.\ \ref{sec:notations} and in \cite{DamWer22}.

\section{The generator}\label{sec:generator}

In Theorem \ref{mainTheor} we constructed  the explicit form of the quantities involved in the action of $\Tcal_t$ on the Weyl operators; then, by linearity and weak$^*$-continuity, we obtained the action on the whole $W^*$-algebra $\Nscr$. So, the generator of the semigroup was not needed to determine the semigroup, but it is useful to better understand the physical interactions and to connect the quasi-free case to other situations.

By differentiating \eqref{Tqf} and using \eqref{WtoR}, we get the action of the generator $\Kcal$ of $\Tcal_t$ on the Weyl operators:
\begin{equation*}
\Kcal[W(\xi)]
=\psi(\xi)W(\xi)+\frac \rmi 2 \left\{ R^\T   Z\xi, W(\xi)\right\}.
\end{equation*}
Then, by dividing a ``Gaussian'' part from a ``compensated jump'' part, we can write this generator as
\begin{subequations}\label{K1+K_2xi}
\begin{equation}\label{K1W}
\Kcal= \Kcal_1 +\Kcal_2, \qquad \Kcal_1[W(\xi)]=\left(\rmi \alpha^\T \xi -\frac 12 \,\xi^\T A \xi\right)W(\xi)+\frac \rmi 2 \left\{ R^\T  Z\xi, W(\xi)\right\},
\end{equation}
\begin{equation}\label{K2W}
\Kcal_2[W(\xi)]=\int_{\Xi}\nu(\rmd \eta)\left(\rme^{\rmi {\eta}^\T \xi}-1-\rmi \ind_{\Sbb}(\eta) {\eta}^\T \xi\right)W(\xi).
\end{equation}
\end{subequations}
We recall that $(A,\, \nu,\,\alpha) $ is the generating triplet of $\psi$ given in Theor.\ \ref{mainTheor}, and
$ Z$ is the real $(2n+s)\times (2 n+s)$-matrix generating the semigroup $S_t$ \eqref{Ssemig}.
Let us stress that the above decomposition is not unique, as there are equivalent ways of writing the compensator in the \LevKhi\  formula \cite[Remark 8.4]{Sato99};  a change in the compensator amounts to a change in $\alpha$.
This generator can be extended at least to the linear span of the Weyl operators and expressed by introducing suitable operators on the Hilbert space $\Hscr$.

\begin{remark}\label{rem:block}
In the following we use the block notation with respect to the decomposition $\Xi=\Xi_1\oplus \Xi_0$:
\begin{equation}\label{blocks1}
\alpha=\begin{pmatrix}\beta \\ \alpha^0\end{pmatrix}, \qquad A=\begin{pmatrix}A^{11} & A^{10} \\ A^{01}& A^{00}\end{pmatrix}, \qquad Z=\begin{pmatrix}Z^{11} & Z^{10} \\ Z^{01}& Z^{00}\end{pmatrix};
\end{equation}
recall that $A$ is a real symmetric matrix. With this notation, the matrix $B$ defined in \eqref{posmatr} becomes
\begin{equation}\label{B,Z}
B=\frac 1 2 \,\left(\sigma Z-Z^\T\sigma^\T \right) =\frac 12 \begin{pmatrix}\sigma Z^{11} -{Z^{11}}^\T \sigma^\T & \sigma Z^{10} \\ -{Z^{10}}^\T\sigma^\T& 0\end{pmatrix}= \begin{pmatrix}B^{11}  & B^{10} \\ B^{01}& 0\end{pmatrix}.
\end{equation}
We define also
\begin{equation}\label{D<-Z}
D^{11}=\frac 1 2 \,P_1\left( Z \sigma+\sigma^\T Z^\T \right)P_1 =\frac 12 \left( Z^{11}\sigma +\sigma^\T{Z^{11}}^\T  \right).
\end{equation}
By inverting the previous equations with respect to $Z$, we get
\begin{equation}\label{Z:D+B}
Z^{11}= D^{11}\sigma^\T-\sigma B^{11}, \qquad Z^{10}=2\sigma^\T B^{10}.
\end{equation}
\end{remark}

\begin{proposition}\label{prop:diffgen} When $a$ is in the linear span of the Weyl operators and $f$ is bounded and twice differentiable, the generator \eqref{K1+K_2xi} can be written in the following form: $\Kcal= \Kcal_1 +\Kcal_2$,
\begin{subequations}\label{K1+K_2diff}
\begin{multline}\label{K1(x)}
\Kcal_1[a\otimes f](x)= \frac 12\, a\sum_{i,j=1}^{s} A^{00}_{ij}\,\frac{\partial^2 f(x)}{\partial x_i \partial x_j}+\rmi\sum_{i=1}^{2n} \sum_{j=1}^{s}R_i a\left[\sigma^\T\left( A^{10}- \rmi B^{10}\right)\right]_{ij}  \frac{\partial f(x)}{\partial x_j}
\\ {}  -\rmi \sum_{i=1}^{s} \sum_{j=1}^{2n} \frac{\partial f(x)}{\partial x_i} \left[\left( A^{01} - \rmi B^{01}\right)\sigma\right]_{ij} aR_j
+  \sum_{i,j=1}^{2n}\left[\sigma^\T\left( A^{11} -\rmi B^{11}\right)\sigma\right]_{ij}\left(R_iaR_j -\frac 12 \left\{R_iR_j,a\right\}\right)f(x)
\\ {} +a\sum_{j=1}^s \alpha^0_j\, \frac{\partial f(x)} {\partial x_j}+ a \sum_{i,j=1}^s x_i Z^{00}_{ij} \frac{\partial f(x)}{\partial x_j}+\rmi\left[ H_{\rm q}+H_x,\, a\right]f(x) ,
\end{multline}
\begin{equation}\label{Hamilq}
H_{\rm q}=\beta^\T \sigma P_1R+\frac 12 \,R^\T P_1  D^{11}  P_1 R,
\end{equation}
\begin{equation}\label{Hamilx}
H_x= x^\T P_0 ZP_1\sigma R=\sum_{i=1}^s \sum_{j=1}^{2n} x_i Z^{01}_{ij}\left(\sigma R\right)_j,
\end{equation}
\begin{multline}\label{K2(x)}
\Kcal_2[a\otimes f](x)=\int_{\Xi}\nu(\rmd \eta)\biggl\{f(x+y)W_1(\sigma\zeta')^\dagger a W_1(\sigma\zeta') - f(x)a
\\ {}- \ind_{\Sbb}(\eta) \biggl( \rmi f(x) \left[{\zeta'}^\T \sigma R, \,a\right]+ a\sum_{j=1}^s y_j\frac{\partial f(x) }{\partial x_j} \biggr)\biggr\}, \qquad \eta=\begin{pmatrix}\zeta' \\ y\end{pmatrix}.
\end{multline}
\end{subequations}
\end{proposition}
Recall that $B$, $D$ and $Z$ are connected by \eqref{D<-Z} and \eqref{Z:D+B}.
\begin{proof}
By Eqs.\ \eqref{translations} and \eqref{WtoR}, we get the correspondence rules: $(P_1\xi)_i=\zeta_i \to [(\sigma R)_i, \bullet ]$, \ $(P_0\xi)_i=k_i \to -\rmi \,\frac{\partial \ }{\partial x_i}$. By applying these rules to $\Kcal_1$ \eqref{K1W}, we obtain \eqref{K1(x)}-\eqref{Hamilx} by direct computations.

By using \eqref{classicaltransl}, \eqref{W*WW}, \eqref{WtoR}, and ${\eta}^\T= ( {\zeta'}^\T, y^\T)$ we get
\[
\rme^{\rmi {\eta}^\T \xi}W(\xi)=\rme^{\rmi y^\T k}W_0(k) \rme^{\rmi {\zeta'}^\T  \zeta}W_1(\zeta) = W_1(\sigma\zeta')^\dagger W_1(\zeta)W_1(\sigma\zeta')\otimes \Rcal_y[W_0(k)],
\]
\[
\left[\rmi {\eta}^\T\xi W(\xi)\right](x)=\rmi \left[{\zeta'}^\T \sigma R, W_1(\zeta)\right]W_0(k)(x)+ W_1(\zeta)\sum_{j=1}^s y_j\frac{\partial \ }{\partial x_j}\, W_0(k)(x);
\]
the classical translation operator $\Rcal_y$ is defined in Sec.\ \ref{sec:notations}. Then, we have \eqref{K2(x)}.
\end{proof}

Eventually, the domain of the generator can be extended by weak$^*$-closure. It is possible to check that the structure of this generator is analogous to the structure of  \cite[(4.39)]{BarH95}, while in that reference only bounded operators on the Hilbert space were allowed.

The first four addends in \eqref{K1(x)} form a hybrid dissipative contribution and they are connected  by the positivity condition \eqref{posmatr}; indeed we have
\begin{multline}\label{ABZglobal}
0\leq \left(\sigma^\T+P_0\right)\left(A-\rmi B\right)\left(\sigma+P_0\right)=\begin{pmatrix}\sigma^\T\left( A^{11} -\rmi B^{11}\right)\sigma&
\sigma^\T\left( A^{10}- \rmi B^{10}\right)\\ \left( A^{01} - \rmi B^{01}\right)\sigma & A^{00}\end{pmatrix}
\\
{}=\begin{pmatrix}\sigma^\T A^{11}\sigma  +\frac \rmi2\left(\sigma^\T {Z^{11}}^\T- Z^{11}\sigma\right)&
\sigma^\T A^{10}- \frac\rmi 2 \, Z^{10}\\ A^{01}\sigma +\frac \rmi2\, {Z^{10}}^\T & A^{00}\end{pmatrix}.
\end{multline}

\subsection{The reduced dynamics}\label{sec:reddyn}
Let us start from the reduced dynamics of the quantum component. By setting $f(x)=1$ in the expression \eqref{K1+K_2diff} of the generator we get
\begin{equation}\label{Ka1}
\Kcal[a\otimes 1](x)= \Lcal_{\rm q}[a]+\rmi \left[H_x,\, a\right],
\end{equation}
\begin{multline}\label{L_q}
\Lcal_{\rm q}[a]= \rmi\left[ H_{\rm q}, a\right]+ \sum_{i,j=1}^{2n}\left[\sigma^\T\left( A^{11} -\rmi B^{11}\right)\sigma\right]_{ij}\left(R_iaR_j -\frac 12 \left\{R_iR_j,a\right\}\right)
\\ {}+
\int_{\Xi}\nu(\rmd \eta)\left\{W_1(\sigma\zeta')^\dagger a W_1(\sigma\zeta') - a
- \ind_{\Sbb}(\eta)  \rmi  \left[{\zeta'}^\T \sigma R, \,a\right] \right\}, \qquad \zeta'=P_1\eta.
\end{multline}
The last term in \eqref{Ka1} can be seen as a random Hamiltonian evolution, because the classical variables $x_j$, $j=1,\ldots,s$, appear in $H_x$, which is defined in \eqref{Hamilx}. The dynamics of the quantum component alone is not autonomous, unless $Z^{01}$ vanishes; we can say that $Z^{01}$ controls  the information flow from the classical system to the quantum one. As we shall see in Sec.\ \ref{sec:qds}, $\Lcal_{\rm q}$ can be written in the usual Lindblad form; the jump term and the term involving $A^{11}$ are of dissipative type.

We consider now the reduced dynamics of the classical component. By setting $a=\openone$ in the expression \eqref{K1+K_2diff} of the generator and by taking into account \eqref{ABZglobal} and that $A$ is a symmetric matrix and $B$ an antisymmetric one, we get
\begin{equation}\label{K1f}
\Kcal[\openone\otimes f](x)= \Kcal_{\rm cl}[ f](x)+ \sum_{i=1}^s \sum_{j=1}^{2n} R_jZ^{10}_{ji} \,\frac{\partial f(x)}{\partial x_i},
\end{equation}
\begin{multline}\label{K_cl}
\Kcal_{\rm cl}[f](x)=  \sum_{j=1}^s \alpha^0_j\, \frac{\partial f(x)} {\partial x_j} + \sum_{i,j=1}^s x_i Z^{00}_{ij} \frac{\partial f(x)}{\partial x_j}+\frac 12\sum_{i,j=1}^{s}A^{00}_{ij}\,\frac{\partial^2 f(x)}{\partial x_i \partial x_j}
\\ {}+
\int_{\Xi}\nu(\rmd \eta)\biggl\{f(x+y) - f(x)- \ind_{\Sbb}(\eta) \sum_{j=1}^s y_j\frac{\partial f(x) }{\partial x_j} \biggr\}, \qquad y=P_0\eta.
\end{multline}
The last term in \eqref{K1f} contains the quantum operators $R_j$; so, also the dynamics of the classical component alone is not autonomous, unless $Z^{10}$ vanishes. We can say that $Z^{10}$ controls  the information flow from the quantum system to the classical one. Again the jump term and the term with the second derivatives are of dissipative type; we shall discuss their role in Sec.\ \ref{classical}. The classical component $\Kcal_{\rm cl}$ of the generator is exactly the generator given in \cite[(1.1)]{SatoY84}, for the so called \emph{processes of Ornstein-Uhlenbeck type}.

\subsection{The classical-quantum interaction }\label{sec:gen+int}
To clarify the meaning of the various terms in the generator \eqref{K1+K_2diff}, it is useful to rewrite it   in a way which puts in evidence the quantum-classical interaction terms:
\begin{subequations}\label{interactions}
\begin{gather}\label{red+int}
\Kcal[a\otimes f](x)=\Lcal_{\rm q}[a]f(x)+ a\Kcal_{\rm cl}[ f](x) +\sum_{l=1}^4\Kcal_{\rm int}^l[a\otimes f](x),
\\ \label{K12int1}
\Kcal_{\rm int}^1[a\otimes f](x)=\rmi \left[H_x,\, a\right]f(x), \qquad
\\ \label{K12int2}
\Kcal_{\rm int}^2[a\otimes f](x)= \frac 12 \sum_{i=1}^s \sum_{j=1}^{2n} \{a,R_j\}Z^{10}_{ji} \,\frac{\partial f(x)}{\partial x_i},
\\  \label{K3int}
\Kcal_{\rm int}^3[a\otimes f](x)=\sum_{i=1}^s \sum_{j=1}^{2n} \rmi[\left(\sigma R\right)_j,a ] A^{10}_{ji}\,\frac{\partial f(x)}{\partial x_i},
\\ \label{K4int}
\Kcal_{\rm int}^4[a\otimes f](x)=
\int_{\Xi}\nu(\rmd \eta)\left(f(x+P_0\eta)-f(x)\right)\left(W_1(\sigma P_1\eta)^\dagger a W_1(\sigma P_1\eta)-a\right).
\end{gather}
\end{subequations}

The interaction $\Kcal_{\rm int}^1$ \eqref{K12int1} involves the random Hamiltonian $H_x$, defined in \eqref{Hamilx}; its role has been discussed in Sec.\ \ref{sec:reddyn}: it represents a kind of force exerted on the quantum system by the classical one. On the other side, the interaction $\Kcal_{\rm int}^2$ \eqref{K12int2} represents some action of the quantum system on the classical one.
Note  that $Z^{10}\equiv 2\sigma^\T B^{10}$ is involved in the positivity condition \eqref{posmatr}, \eqref{ABZglobal}; in some sense this interaction term injects some quantum uncertainty into the classical output.

The interaction term $\Kcal_{\rm int}^3$ \eqref{K3int} has a peculiar structure, as it vanishes either when the reduced classical dynamics is considered ($a=\openone$), either when the reduced quantum dynamics is considered ($f(x)=1$). Also this term is involved in the positivity condition \eqref{posmatr}, \eqref{ABZglobal}.

Finally, $\Kcal_{\rm int}^4$ \eqref{K4int} represents the interaction contained in the jump part;
also this term vanishes either from the reduced classical dynamics, either from the reduced quantum dynamics. This is a property of the quasi-free character of the dynamics, because the Lindblad operators involved in the jump part are proportional to unitary operators; it could be seen that this fact makes the probability law of the inter-jump time independent from the quantum state \cite{WisM10,BarB91}.

\subsubsection{The role of the dissipative terms}
It is interesting to see what happens in the case of no-dissipation in the purely quantum component of the dynamics, i.e.\ in $\Lcal_{\rm q}$ given by \eqref{L_q}. Firstly, we have to take $A^{11}=0$;  by the positivity condition \eqref{posmatr} and the structure \eqref{blocks1}, \eqref{B,Z}, \eqref{Z:D+B}, this gives also the vanishing of other terms:
\begin{equation*}
A^{11}=0 \qquad \Rightarrow \qquad \begin{matrix} A^{01}={A^{10}}^\T=0, \quad B^{01}=-{B^{10}}^\T=0, \\   B^{11}=0, \quad Z^{10}=0, \quad  Z^{11}=D^{11}\sigma^\T .\end{matrix}
\end{equation*}
To have no dissipation in $\Lcal_{\rm q}$, also the compensated jump term in \eqref{L_q} has to vanish. These conditions imply the vanishing of the interaction terms $\Kcal_{\rm int}^2$,  $\Kcal_{\rm int}^3$, $\Kcal_{\rm int}^4$; the only classical-quantum interaction which survives is the Hamiltonian contribution $\Kcal_{\rm int}^1$, linked to the flux of information from the classical system to the quantum one. So, we have the vanishing of all the terms which can transfer information from the quantum  to the classical system, which in principle could be observed without disturbance; again, also for the introduced hybrid Markovian quasi-free dynamics we have that no information can be extracted from the quantum system without dissipation.

Also the specular statement holds for the hybrid quasi-free dynamics: if the classical component extracts information
from the quantum one, it acquires necessarily a stochastic behaviour. Indeed, if we require the vanishing of the diffusive term in $\Kcal_{\rm cl}$ \eqref{K_cl} we have:
\begin{equation*}
A^{00}=0 \qquad \Rightarrow \qquad \begin{matrix} A^{01}={A^{10}}^\T=0, \quad B^{01}=-{B^{10}}^\T=0, \\ Z^{10}=0, \qquad \Kcal_{\rm int}^2=0, \quad  \Kcal_{\rm int}^3=0.\end{matrix}
\end{equation*}
To have no dissipation in $\Kcal_{\rm cl}$, also the compensated jump term in \eqref{K_cl} has to vanish; then, we have also $\Kcal_{\rm int}^4=0$. Again, only the interaction term $\Kcal_{\rm int}^1$, which  contains $H_x$, survives: a deterministic classical system can act only as a kind of control on the quantum system.

Before discussing the meaning and the properties of the general hybrid dynamics, it is useful to see the case of a pure quantum dynamical semigroup (Sec.\ \ref{sec:qds}) and the pure classical case (Sec.\ \ref{classical}). From the pure quantum case we can have some hints on the possible interactions among quantum components and with the external environment. In the pure classical case we can see that the dynamical semigroup does not give only the state at time $t$, the probability distribution at a single time, but a whole stochastic process can be constructed with all multi-time probabilities. This idea can be extended to the hybrid case and we can interpret the observation of the classical component as a monitoring in continuous time of the quantum system.

\section{Quasi-free quantum dynamical semigroups}\label{sec:qds}

By taking $s=0$ in Sec.\ \ref{Sec:qfSem} we obtain the most general quasi-free quantum dynamical semigroup. In this case we have   $S_t=\rme^{Z^{11}t}$, where $Z^{11}=Z$ is a generic $2n\times 2n$ real matrix; moreover, $A^{11}=A$ and the positivity condition \eqref{posmatr}  reduces to
\begin{equation}\label{A11+iB11}
A^{11}\pm \rmi B^{11}\geq 0, \qquad B^{11}=\frac 12 \left(\sigma Z^{11} -{Z^{11}}^\T \sigma^\T\right).
\end{equation}
Then, Eq.\ \eqref{Tqf} becomes
\begin{equation}\label{Tstrucq}
\Tcal_t[W_1(\zeta)]=\exp\left\{\int_0^t\rmd \tau \,\psi_1(S_\tau \zeta)\right\} W_1(S_t \zeta), \qquad \forall \zeta\in \Rbb^{2n}, \quad \forall t\geq 0;
\end{equation}
$\psi_1(\zeta)$ is obtained by particularizing \eqref{formpsi}, \eqref{propnu}:
\begin{equation}\label{psi1}\begin{split}
&\psi_1(\zeta)= \rmi \beta^\T \zeta -\frac 12 \,\zeta^\T A^{11} \zeta + \int_{\Rbb^{2n}}\nu_1(\rmd \zeta')\left(\rme^{\rmi {\zeta'}^\T  \zeta}-1-\rmi \ind_{\Sbb_1}(\zeta') {\zeta'}^\T \zeta\right),
\\
&\Sbb_1=\left\{ \zeta\in \Rbb^{2n}: \abs\zeta<1\right\}, \qquad \int_{\Sbb_1\setminus 0} \abs{\zeta'}^2\nu_1(\rmd \zeta')<+\infty, \qquad \nu_1(\{1\leq \abs \zeta' <+\infty\})<+\infty.\end{split}
\end{equation}

In the general hybrid case, we can obtain the dynamics of the quantum subsystem alone by taking $\xi^\T=(\zeta^\T,0)$ in \eqref{Tqf}. As seen in Sec.\ \ref{sec:reddyn}, the reduced quantum dynamics is not a semigroup and it depends on the dynamics of the classical system. The reduced quantum dynamics is a semigroup if and only if we have
\begin{equation*}
Z^{01}=0.
\end{equation*}
Under this assumption, we get $\psi(P_1\zeta)=\psi_1(\zeta)$ by identifying the generating  triplet:
\[
A^{11}=P_1AP_1, \qquad \nu_1(E)=\int_\Xi \nu(\rmd \eta)\, \ind_{E}(P_1\eta),
\qquad
\beta=P_1\alpha+ \int_{\{\eta\in \Xi: \abs{P_1\eta}^2<1, \,\abs{\eta}\geq 1\}}\nu(\rmd \eta)\,P_1\eta.
\]
A similar reduction problem is considered in \cite[Prop.\ 11.10]{Sato99}, in the pure classical case.

The generator $\Lcal_{\rm q}$ of the most general quasi-free quantum dynamical semigroup on $\Bcal(\Hcal)$ takes the form \eqref{L_q}, \eqref{Hamilq}.  We can say that $\Tcal_t$, given by \eqref{Tstrucq} and \eqref{psi1}, solves the Markovian quantum master equation (in the Heisenberg description) $\dot \Tcal_t[a]=\Tcal_t\big[\Lcal_{\rm q}[a]\big]$. By separating the Gaussian part from the compensated jump part in $\Lcal_{\rm q}$, we can write
\begin{subequations}\label{genL}
\begin{gather}
\Lcal_{\rm q}= \Lcal_1 +\Lcal_2, \\ \Lcal_1[a]=\rmi\left[H_{\rm q},a\right]+ \sum_{i,j=1}^{2n}\left[\sigma\left(A^{11}-\rmi B^{11}\right)\sigma^\T\right]_{ij}\left(R_iaR_j -\frac 12 \left\{R_iR_j,a\right\}\right),
\\ \label{ABD}
H_{\rm q}=\frac 12\sum_{i,j=1}^{2n} R_i D_{ij}^{11} R_j+ \sum_{i,j=1}^{2n}\beta_i\sigma_{ij}R_j, \quad D^{11}=\frac 12 \left(Z^{11}\sigma + \sigma^\T {Z^{11}}^\T\right), \quad
\beta_j\in \Rbb,
\\ \label{Ljump}
\Lcal_2[a]=\int_{\Rbb^{2n}}\nu_1(\rmd \zeta')\left( W_1(\sigma\zeta')^\dagger a W_1(\sigma\zeta')-a-\rmi \ind_{\Sbb_1}(\zeta')\left[{\zeta'}^\T\sigma R, a\right]\right).
\end{gather}
\end{subequations}

Let us note that $\Lcal_{\rm q}$ is an unbounded generator. It is known that many analytical problems are involved in the construction of the semigroup from a formal unbounded generator \cite{ChebF98,Hol95,SieHW17}; however, in the quasi-free case the semigroup is defined directly by its action on the Weyl operators \eqref{Tstrucq}, \eqref{psi1}, and these problems do not arise.

\begin{remark}\label{rem:long} By the positivity condition in \eqref{A11+iB11} we can write the matrix $\sigma\left(A^{11}-\rmi B^{11}\right)\sigma^\T $ in terms of its eigenvalues and eigenvectors:
\[
\left[\sigma\left(A^{11}-\rmi B^{11} \right)\sigma^\T\right]_{ji}=\sum_{l=1}^{2n}\lambda_l a_l^j\, \overline{a_l^i}, \qquad \lambda_l\geq 0, \qquad \sum_{i=1}^{2n} \overline{a_l^i}\, a_{l'}^i= \delta_{l l '}.
\]
Then, the Gaussian contribution can be written in a Lindblad-like form as
\[
\Lcal_1[a]=\rmi\left[H_{\rm q},a\right]+\sum_{l=1}^{2n}\left(L_{l}^\dagger aL_{l} -\frac12\left\{L_{l}^\dagger L_{l},a\right\} \right), \qquad L_l=\sqrt{\lambda_l} \sum_{i=1}^{2n} \overline{a_l^i}\, R_i.
\]
Similarly, by using $-a= -\frac 12 \left\{W_1(\sigma \zeta')^\dagger  W_1(\sigma\zeta'),a\right\}$ inside the integral in \eqref{L_q}, we see that the compensated jump contribution can be written as an integral of Lindblad-like expressions.
\end{remark}

\begin{remark}\label{rem:BDZ} If we start from the operator expression \eqref{genL}, with the matrices $B^{11}$ and $D^{11}$ given, we can obtain the action of the generator on the Weyl operators by  taking $Z^{11}$  from \eqref{Z:D+B}, which gives it
in terms of $B^{11}$ and $D^{11}$.

Note that $\beta$ and $D^{11}$ come out from the Hamiltonian contribution, while $A^{11}$ and $B^{11}$ come out from the dissipative part of the Gaussian contribution. Moreover, the compensated jump part does not contribute to the ``deterministic'' dynamics $S_t=\rme^{Z^{11} t}$.
\end{remark}

\subsection{The state dynamics} As already recalled in Sec.\ \ref{sec:Wigf} for the hybrid case,
a state $\rho$ is completely determined by its \emph{characteristic function} $\chi_\rho(\zeta)= \Tr \left\{W_1(\zeta)\rho\right\}$. The anti-Fourier transform of $\chi_\rho(\zeta)$ is the \emph{Wigner function} of $\rho$; when the Wigner function is non-negative, it is a probability density and we say that the corresponding state is ``classical''.

Let ${\Tcal}_{t*}$ be the preadjoint of ${\Tcal}_{t}$; then, $\rho_t={\Tcal}_{t*}[\rho_0]$ is the state at time $t$.  The characteristic function of the state at time $t$ is
\begin{equation}\label{pureq}
\chi_{\rho_t}(\zeta)=\Tr\left\{{\Tcal}_{t}[W_1(\zeta)]\rho_0\right\}=\exp\left\{\int_0^t\rmd \tau \,\psi_1(\rme^{Z^{11}\tau} \zeta)\right\} \chi_{\rho_0}(\rme^{Z^{11} t}\zeta).
\end{equation}
Note that the exponential pre-factor does not depend on the initial quantum state; by Remark \ref{+intpsi} this factor is the characteristic function of a probability law (which is non-negative). Therefore, when $\rho_0$ is a classical state, $\rho_t$ too turns out to be a classical state. When $\rho_0$  is non-classical, the product structure in the characteristic function implies a convolution structure in the Wigner function, which tends to eliminate non-classicality.

By using \eqref{translations} and \eqref{W*WW}, from the generator structure \eqref{genL} we get easily
\begin{gather*}
W_1( z)^\dagger\Lcal_1[W_1(z)W_1(\zeta)W_1( z)^\dagger]W_1( z)=\Lcal_1[W_1(\zeta)] -\rmi z^\T Z^{11}\sigma \zeta W_1(\zeta),
\\
W_1( z)^\dagger\Lcal_2[W_1( z)W_1(\zeta)W_1( z)^\dagger]W_1(z)=\Lcal_2[W_1(\zeta)].
\end{gather*}
So, we have that the jump component $\Lcal_2$ is \emph{Weyl covariant} \cite{Hol95,Hol96}. The Gaussian component $\Lcal_1$ enjoys the same covariance property only if $Z^{11}=0$, which means $D^{11}=0$ and $B^{11}=0$; then, $\Lcal_1$ takes the form
\begin{equation}\label{L1cov}
\Lcal_1[a]=\rmi\left[\beta^\T\sigma R,a\right]-\frac12
\sum_{i,j=1}^{2n} \left(\sigma A^{11}\sigma^\T\right)_{ji}\left[R_j,\left[ R_i ,a\right]\right].
\end{equation}
By Proposition 1 in \cite{Hol96} the most general Weyl covariant quantum dynamical semigroup turns out to be quasi-free and, indeed, the adjoint of the generator (2.12) in Proposition 2 of \cite{Hol96} coincides with the generator defined by the sum of \eqref{Ljump} and \eqref{L1cov}. The same generator was already found in \cite[Theor.\ 1]{Hol95}.

\begin{remark}[Existence of a unique equilibrium state]\label{rem:equi} If we have $\abs {\rme^{Z^{11}t}\zeta}\leq \rme^{-\varkappa t} \abs{\zeta}$, $\forall \zeta\in \Rbb^{2n}$, $\forall t\geq 0$, with $\varkappa>0$, and $\int_{\abs \zeta>1}\ln\abs \zeta \nu_1(\rmd \zeta)<+\infty$, by the results of Sec.\ \ref{sec:equil}, we have
\begin{equation*}
\chi_{\rho_{\rm eq}}(\zeta)=\lim_{t\to +\infty}\chi_{\rho_t}(\zeta)=\exp\left\{\int_0^{+\infty}\rmd \tau \,\psi_1(\rme^{Z^{11}\tau} \zeta)\right\}.
\end{equation*}
The equilibrium state has a non-negative Wigner function, given by an infinitely divisible probability distribution. Note that, by \eqref{Atpos}, we have
$
A_\infty \pm \frac\rmi 2 \,\sigma\geq 0$, and the equilibrium state is not singular.
\end{remark}

By the property \eqref{WtoR} of the Weyl operators, we get that the quantum means  are given by
\begin{equation}\label{qmeans}\begin{split}
&\langle R\rangle_t={S_t}^\T\langle R\rangle_0+\int_0^t \rmd \tau \,{S_\tau}^\T\tilde \beta, \qquad  \frac{\rmd \ }{\rmd t}\langle R\rangle_t={Z^{11}}^\T\langle R\rangle_t+\tilde \beta,
\\
&\langle R_j\rangle_t := \Tr\left\{R_j\rho_t\right\},  \qquad \tilde \beta:=\beta+ \int_{\abs{\zeta}\geq 1}\zeta\nu_1(\rmd \zeta).
\end{split}\end{equation}
The quantum means exist when $\langle R_j\rangle_0$ is finite and  the integral in the definition of $\tilde \beta$ exists.

\subsection{Examples}

Quasi-free quantum dynamical semigroups, with a non-vanishing jump part, already appeared in the literature, mainly under symmetry requirements; for instance, the Galilean covariant evolutions considered in \cite[Sec.\ III]{Hol96} turn out to be quasi-free semigroups.
When less symmetry is considered, the resulting semigroup is not necessarily quasi-free, but jump structures are included. For instance, in \cite[Theor.\ 2]{Hol95} the most general generator of a space-translation covariant semigroup is obtained; in this case to ask the semigroup to be quasi-free is a restriction.

Under the name of \emph{quantum linear Boltzmann equation}, translation covariant quantum dynamical semigroups have been used to study the motion of a test particle in a gas and the so called dynamical decoherence \cite{Vac08,VacH09}.

\subsubsection{Quasi-free quantum linear Boltzmann equation}
As a first example, we give here the quasi-free version of the quantum linear Boltzmann equation studied in \cite{VacH09}.
We take $n=3$ and consider the case of covariance under spatial translations and rotations.

We use a notation with $3\times 3$-blocks (position variables and momentum variables) and we take
\[
{Z^{11}}^\T=\begin{pmatrix}0& \frac \openone m\\ 0 & -\gamma\openone \end{pmatrix}, \qquad A^{11}=\begin{pmatrix}a_1\openone & a_3\openone \\ a_3\openone& a_2\openone\end{pmatrix}, \qquad \beta=0,
\]
\[
a_1\geq 0, \qquad a_2 \geq 0, \qquad a_1 a_2 \geq a_3^{\;2}+ \frac {\gamma^2}4, \qquad m>0, \qquad \gamma> 0;
\]
this gives also
\[
D^{11}=\begin{pmatrix}0& \frac\gamma 2 \,\openone\\ \frac\gamma 2 \,\openone& \frac \openone m\end{pmatrix}, \qquad B^{11}=\frac \gamma 2 \begin{pmatrix} 0& - \openone\\  \openone& 0 \end{pmatrix},\qquad A^{11}\pm \rmi B\geq 0, \] \[
S_t= \begin{pmatrix} \openone& 0\\ \frac{1-\rme^{-\gamma t}}{m\gamma}\,\openone &\rme^{-\gamma t}\openone  \end{pmatrix}.
\]
Then, we take the measure $\nu_1$ to be rotationally invariant; as pointed out in \cite{Hol96} after Eq.\ (3.14), the compensator formally vanishes and we can write
\[
\psi_1(\zeta)= -\frac 12 \,\zeta^\T A^{11} \zeta + \lim_{\epsilon \downarrow 0}\int_{\abs{\zeta'}>\epsilon}\nu_1(\rmd \zeta')\left(\rme^{\rmi \left(\zeta^\T\zeta'\right)}-1\right).
\]
Finally, we ask that the quantum expectations $\langle R_j\rangle_t$ exist.

In this example  there is not a stationary state at large times and, moreover, the memory of the initial state is not lost. However, the quantum distribution of the momentum reaches a stationary expression and looses any  contribution of the initial state. Indeed, by using the improper eigenfunctions $| p\rangle$ of the momentum operators, the momentum density at time $t$ is given by
\[
\langle p|\rho_t| p\rangle= \frac 1 {(2\pi)^3}\int_{\Rbb^3} \rmd k\, \rme^{-\rmi  k^\T p}\chi_{\rho_t}(0,\,k) .
\]
Then, we get
\[
\lim_{t\to+\infty}\chi_{\rho_t}(0,\,k)=\exp\left\{-\frac{a_2\abs{k}^2}{4\gamma} +\int_{\Xi}\nu_1(\rmd \zeta')\int_0^{+\infty}\rmd \tau\left(\exp\left\{\rmi (0,k^\T) \zeta'\rme^{-\gamma\tau}\right\}-1\right)\right\},
\]
which is the characteristic function of an infinitely divisible distribution. Let us note that  $\gamma>0$ implies $a_2>0$ and the Gaussian contribution cannot vanish.

\subsubsection{An optomechanical system}\label{sec:optomec}
We consider now a quantum harmonic oscillator with dissipation. The choice of the dynamics depends on the physical system we want to describe, an optical mode in a cavity or a micro-mechanical system \cite{Vac02,Bar16,BarG21}. Here we take as example
an oscillating micro-mirror, hit by free photons (not in a cavity) and damped by a phonon bath \cite{Bar16,BarG21}.  The interaction between the mechanical mirror and the photons is by the radiation pressure. The micro-mirror is a mechanical system; at least at the level of the quantum means, the forces must appear only in the equation for the momentum.

Let us take $n=1$; we assume the existence of the quantum expectations of $Q$ and $P$. By the previous considerations we want to have $\frac{\rmd\ }{\rmd t}\langle Q\rangle_t= \Omega \langle P\rangle_t$; by \eqref{qmeans} we must have ${Z^{11}}^\T_{11}=0$.  Moreover, we choose the underdamped case; the dynamical matrix ${Z^{11}}^\T$ and its eigenvalues are given by the expressions
\begin{equation}\label{lambda_pm}\begin{split}
{Z^{11}}^\T = \begin{pmatrix}0 & \Omega\\ -\Omega& -\gamma \end{pmatrix}, \qquad \gamma>0,\qquad \Omega>0, \qquad \Omega^2> \frac{\gamma^2}4 , \\ \lambda_\pm=-\frac \gamma 2 \pm \rmi \omega, \qquad \omega=\sqrt{\Omega^2-\frac{\gamma^2}4}>0.
\end{split}\end{equation}
Then, we have
\[
D^{11}=\frac 12 \left(Z^{11}\sigma +\sigma^\T{Z^{11}}^\T\right)=\begin{pmatrix} \Omega& \frac\gamma 2\\ \frac\gamma 2& \Omega \end{pmatrix}, \qquad B^{11} =\frac12 \left( \sigma Z^{11}- {Z^{11}}^\T\sigma^\T \right)= - \frac\gamma 2\,\sigma.
\]
The matrix $A^{11}$ cannot have vanishing eigenvalues because of the positivity condition \eqref{A11+iB11}, which gives
\[
A^{11}=\begin{pmatrix}a_1 & a_3 \\ a_3& a_2\end{pmatrix}, \qquad a_1>0, \qquad a_2>0, \qquad a_1a_2\geq a_3^{\,2}+\frac {\gamma^2}4.
\]
The jump contribution has to give rise only to the radiation pressure force (momentum kicks); so, we take
the measure $\nu_1$ to be concentrated on the second component:
\[
\int_{\Rbb^{2}}\nu_1(\rmd \zeta')f(\zeta_1',\zeta_2')=\int_{\Rbb}m_1(\rmd v)f(0,v).
\]
Apart from damping and radiation pressure, no other force must be present; so, we take also
$\beta_1=0$ and $\beta_2=\int_{\{\abs v <1\}} m_1(\rmd v)$, which we assume to exists.  Then we have
\[
\tilde \beta_1=0, \qquad \tilde\beta_2= \int_{\Rbb}v m_1(\rmd v),
\qquad
\psi_1(\zeta)=-\frac 12\, \zeta^\T A^{11} \zeta+\int_{\Rbb}m_1(\rmd v)\left(\rme^{\rmi v \zeta_2} - 1 \right).
\]

The  generator \eqref{genL} reduces to
\[
\Lcal_{\rm q}[a]=\rmi\left[H_{\rm q},a\right]+\sum_{i,j=1}^{2} \left(\sigma A^{11}\sigma^\T+ \frac{\rmi \gamma}2\,\sigma \right)_{ij}\left(R_ia R_j -\frac12\left\{R_iR_j,a\right\}\right)
+ \int_{\Rbb}m_1(\rmd v)\left( \rme^{-\rmi v Q} a \rme^{\rmi v Q}- a\right),
\]
\[
R_1=Q, \quad R_2=P, \qquad H_{\rm q}=\frac 12\sum_{i,j=1}^{2} R_i D_{ij}^{11} R_j= \frac\Omega 2 \left(Q^2+P^2\right)+\frac \gamma 4 \left\{Q,P\right\}.
\]
An interesting point is that the damping $\gamma$ comes out from a combined action of the Hamiltonian term (through $D^{11}$) and of the dissipative Gaussian term (through $B^{11}$).

The dynamical evolution $S_t^\T$ is the same as for a classical damped oscillator:
\begin{equation}\label{xi1,2}
S_t^\T=\frac{\rme^{\lambda_+ t}}{2\rmi \omega} \begin{pmatrix} -\lambda_- &\Omega \\ -\Omega & \lambda_+ \end{pmatrix} + \frac{\rme^{\lambda_- t}}{2\rmi \omega} \begin{pmatrix} \lambda_+ &-\Omega \\ \Omega & -\lambda_- \end{pmatrix}.
\end{equation}
Then, the characteristic function of the state at time $t$ is given by \eqref{pureq}. As $S_t$ decays exponentially to $0$ and the first moments are assumed to exist, the equilibrium state exists and does not depend on the initial state:
\begin{equation}\label{eqcf}
\chi_{\rho_{\rm eq}}(\zeta)=\lim_{t\to + \infty} \chi_{\rho_{t}}(\zeta)\\ {}= \exp\left\{\int_0^{+\infty}\rmd \tau \left[ -\frac 12\left( S_\tau\zeta\right)^\T \sigma A^{11} \sigma^\T S_\tau\zeta+\int_{\Rbb}m_1(\rmd v)\left(\rme^{\rmi v \left(S_\tau\zeta\right)_2} - 1 \right) \right]\right\},
\end{equation}
By taking the derivatives of \eqref{eqcf} with respect to $\zeta_1$, $\zeta_2$, the quantum moments of the equilibrium state can be computed. We assumed the existence of the first moments; the existence of higher moments depends on the properties of the measure $m_1$.

\section{The pure classical case}\label{classical}

Let us take $n=0$; then $S_t=\rme^{Z^{00}t}$, where $Z^{00}=Z$ is a generic $s\times s$ real matrix. Moreover, $A^{00}=A$ and condition \eqref{posmatr} reduces to the fact that $A^{00}$ is a real positive-semidefinite $s\times s$-matrix. Now, Eq.\ \eqref{Tqf} becomes
\begin{equation}\label{Tstrucclas}
\Tcal_t[W_0(k)]=\exp\left\{\int_0^t\rmd \tau \,\psi_0(S_\tau k)\right\} W_0(S_t k), \qquad \forall k\in \Rbb^s, \quad \forall t\geq 0,
\end{equation}
where the structure of $\psi_0$ \eqref{formpsi} reduces to the usual LK-formula 
\begin{equation}\label{psicl}\begin{split}
&\psi_0(k)=\rmi k^\T \alpha^0-\frac 12 \,k^\T A^{00} k+\int_{\Rbb^s}\nu_0(\rmd y)\left(\rme^{\rmi y^\T k}-1-\rmi \ind_{\Sbb_0}(y) y^\T k\right) ,
\\
& \forall k\in \Xi_0=\Rbb^{s}, \qquad \alpha^0\in \Rbb^s, \qquad \Sbb_0=\left\{ y\in \Rbb^{s}: \abs y<1\right\}.
\end{split}\end{equation}
Conditions \eqref{propnu} hold for $\nu_0$.

Similarly to what is done at the beginning of Sec.\ \ref{sec:qds},
the same result can be obtained in the general hybrid case, with the assumption
\begin{equation*}
Z^{10}=0.
\end{equation*}
This gives $S_tP_0=P_0S_tP_0= \rme^{Z^{00}t}P_0$.  Then, equations \eqref{Tqf}, \eqref{formpsi} give \eqref{Tstrucclas}, \eqref{psicl} for $\xi=(0,k)$. Indeed, we get $\psi(P_0k)=\psi_0(k)$ by identifying the generating  triplet:
\[
A^{00}=P_0AP_0, \qquad \nu_0(E)=\int_\Xi \nu(\rmd \eta)\, \ind_E(P_0\eta),
\] \[
\alpha^0=P_0\alpha+ \int_{\{\eta\in \Xi: \abs{P_0\eta}<1,\, \abs{\eta}\geq 1\}}\nu(\rmd \eta)\,P_0\eta;
\]
$E$ is any Borel subset of $\Xi_0$.

The generator of the reduced classical dynamics takes the form \eqref{K_cl}.
Semigroups like $\Tcal_t$ \eqref{Tstrucclas} are well known in the theory of classical stochastic processes; they are semigroups of transition probabilities of time-homogeneous Markov processes \cite{Sato99,App09};  $\dot \Tcal_t[f]=\Tcal_t\big[\Kcal_{\rm cl}[f]\big]$ is a version of the Kolmogorov-Fokker-Planck equation \cite[Secs.\ 3.5.2, 3.5.3]{App09}.

\subsection{The probabilities}\label{sec:clprobs}
In the classical case the initial condition is a probability distribution, having a density: $P_0(E)=\int_{E} \rmd x\, p_0(x)$, \quad $p_0\in L^1(\Rbb^s)$, \quad $p_0\geq 0$, \quad $\int_{\Rbb^s} \rmd x \,p_0(x)=1$. Moreover, the duality form reduces to an expectation:
$\Ebb_0[f]=\int_{\Rbb^s} \rmd x\, p_0(x)f(x)$, \quad $\forall f\in L^\infty(\Rbb^s)$.

\subsubsection{State evolution}
The semigroup ${\Tcal_t}_*$ gives the probability distribution at time $t$: \ $\Ebb_t[f]=\int_{\Rbb^s} \rmd x\, p_t(x)f(x)= \Ebb_0\big[\Tcal_t[f]\big]$, \quad $\forall f\in L^\infty(\Rbb^s)$.
For $f=W_0(k)$ we get $\Ebb_t[W_0]=\int_{\Rbb^s} \rmd x\, p_t(x)\rme^{\rmi k^\T x}=: \hat p_t(k)$, the Fourier transform of $p_t$:
\begin{equation}\label{FTpt}
\hat p_t(k)= \Ebb_0\big[\Tcal_t[W_0(k)]\big]=\int_{\Rbb^s} \rmd y\, p_0(y)\exp\left\{\int_0^t\rmd \tau \,\psi_0(S_\tau k)\right\}\rme^{\rmi y^\T S_t k}=\exp\left\{\int_0^t\rmd \tau \,\psi_0(S_\tau k)\right\}\hat p_0(S_tk).
\end{equation}

As the initial probability distribution has a density, we have that $\hat p_0(S_tk)$ is the characteristic function of a distribution admitting a density. By Remark \ref{+intpsi} and Proposition \ref{prop:tripl=int}, $\exp\left\{\int_0^t\rmd \tau \,\psi_0(S_\tau k)\right\}$ is the characteristic function of an infinitely divisible distribution. Then, the product $\hat p_t(k)$ is the characteristic function of a distribution with density, given by
\begin{equation}\label{1tdens}
p_t(x)= \frac 1 {(2\pi)^s}\int_{\Rbb^s}\rmd k\,\rme^{-\rmi x^\T k}\hat p_t(k)=\int_{\Rbb^s} \rmd y\, p_0(y)\frac 1 {(2\pi)^s}\int_{\Rbb^s}\rmd k\,\rme^{\rmi \left(S_t^\T y-x\right)^\T k}\exp\left\{\int_0^t\rmd \tau \,\psi_0(S_\tau k)\right\}.
\end{equation}

\subsubsection{Louville equation}\label{sec:liou}
Here we show that, in the case of a pure Hamiltonian system, the evolution equation for the probability density \eqref{1tdens} reduces to the Liouville equation.

We consider a classical system with $s=2m$, and canonical coordinates $x_i=q_i$, $x_{m+i}=p_i$, $i=1,\ldots m$. We take the case of no noise: in \eqref{K_cl} we have
\begin{equation}\label{nonoise}
A^{00}=0, \qquad \nu(\rmd \eta)=0.
\end{equation}
We write $p_t(x)= p_{\rm cl}(q,p,t)$ and we use the duality form to find the preadjoint generator acting on $p_{\rm cl}$. By  integration by parts, we shift the derivatives appearing in the generator \eqref{K_cl} from $f$  to $p_{\rm cl}$ and  we get
\begin{multline*}
\frac{\partial p_{\rm cl}(q,p,t)}{\partial t}= -\sum_{i=1}^{2m}Z_{ii}^{ 00}-\sum_{j=1}^m \left(\alpha^0_j\, \frac{\partial p_{\rm cl}(q,p,t)} {\partial q_j}+ \alpha^0_{j+m}\, \frac{\partial p_{\rm cl}(q,p,t)} {\partial p_j}\right)
\\
-\sum_{i,j=1}^m \left[\left(q_iZ^{00}_{ij}+p_iZ^{00}_{i+m,j}\right) \frac{\partial p_{\rm cl}(q,p,t)} {\partial q_j}+ \left( q_iZ^{00}_{i,j+m}+p_iZ^{00}_{i+m,j+m}\right) \frac{\partial p_{\rm cl}(q,p,t)} {\partial p_j}\right].
\end{multline*}
To reduce this equation to an Hamiltonian evolution, we have to take
\begin{equation}\label{conservd}
Z_{i+m,j}^{00}=Z_{j+m,i}^{00},\quad Z_{i,j+m}^{00}=Z_{j,i+m}^{00}, \quad Z_{ji}^{00}=-Z_{i+m,j+m}^{00}, \quad i,j=1,\ldots, m.
\end{equation}
Now, we introduce the classical Hamiltonian
\begin{equation}\label{clHam}
H_{\rm cl}(q,p)=\sum_{i=1}^m\left(\alpha^0_ip_i -\alpha^0_{i+l}q_i\right)+\sum_{i,j=1}^m\left(\frac 12 \, p_iZ_{i+m,j}p_j + \frac 12 \, q_iZ_{i,j+m}q_j+ q_iZ_{i,j}p_j\right);
\end{equation}
the Hamiltonian is of second order because we started with a quasi-free dynamics.
Then, the evolution equation above reduces to
\begin{equation}\label{Liou}
\frac{\partial  p_{\rm cl}(q,p,t)}{\partial t}+ \left\{ p_{\rm cl}(q,p,t),H_{\rm cl}(q,p),\right\}_{\rm cl}=0,
\end{equation}
where $\{\cdot,\cdot\}_{\rm cl}$ denotes the Poisson brackets. Eq.\ \eqref{Liou} is the Liouville equation with Hamiltonian \eqref{clHam}; note that we asked to have a conservative dynamics \eqref{conservd} and no noise \eqref{nonoise}.

\subsubsection{Transition probabilities}\label{sec:conprob}

Let us go back to the general case \eqref{Tstrucclas}, \eqref{psicl}.
The important point of the classical case is that the dynamical semigroup does not give only the probability distribution \eqref{1tdens} at a single time, but also all the multi-time probabilities.

By using \cite[Eq.\ (3.3)]{App09} we can identify the transition probabilities by
\begin{equation}\label{T>condprob}
\Tcal_t[f](x)=\int_{\Rbb^s}f(y) P_t(\rmd y|x).
\end{equation}
By \eqref{1tdens} the characteristic function of this transition probability turns out to be
\begin{equation}\label{FTcond}
\int_{\Rbb^s}\rme^{\rmi k^\T y}P_t(\rmd y|x)=\rme^{\rmi k^\T S_t^\T x}\exp\left\{\int_0^t\rmd \tau \,\psi_0(S_\tau k)\right\};
\end{equation}
again, these transition probabilities are infinitely divisible.
By construction, the transition probabilities satisfy the Chapman-Kolmogorov identity (for the time-homogeneous case) \cite[(10.1)]{Sato99}, \cite[Theor.\ 3.1.5 and (3.6)]{App09}:
\begin{equation}\label{CKeq}
\int_{\Rbb^s} \,P_t(C|y)P_r(\rmd y|x)=P_{t+r}(C|x).
\end{equation}

Having the transition probabilities, also joint probabilities at different times can be introduced (for every choice of a finite number of times); then, by Kolmogorov's extension theorem \cite[Theor.\ 1.8]{Sato99}, it is possible to construct the probability measure for a stochastic process $X(t)$ in continuous time.

\subsection{The process}\label{sec:clprocess}
Let $B(t)$ be the L\'evy process with characteristic function for the increments
\begin{equation}\label{Lprocess}
\Ebb\left[ \exp\left\{ \rmi k^\T\left(B(t+\Delta t)-B(t)\right)\right\}\right]=\exp\left\{\psi_0(k) \Delta t\right\}.
\end{equation}
A L\'evy process has independent and stationary increments, $B(0)=0$, and it is \emph{stochastically continuous} \cite[Sec.\ 1.3]{App09}.
Let us recall that L\'evy processes can be used to define a class stochastic integrals \cite[Sec.\ 4.3]{App09}.

The process $X(t)$ associated with the semigroup $\Tcal_t$ is a Markov process, which satisfies the stochastic differential equation
\begin{subequations}\label{process}
\begin{equation}
\rmd X(t)= {Z^{00}}^\T X(t)\,\rmd t   +\rmd B(t),
\end{equation}
whose solution is
\begin{equation}
X(t)= {S_t}^\T X(0)+\int_0^t {S_{t-\tau}}^\T\rmd B(\tau);
\end{equation}
\end{subequations}
the initial condition $X(0)$ is independent of the process $B$ and it has distribution with density $p_0$. These processes are sometime called of \emph{Ornstein-Uhlenbeck type} \cite{SatoY84}, and they are a sub-class of the processes with independent increments \cite[p.\ 43]{App09},
\cite[Property (1) in Def.\ 1.6]{Sato99}.

To check this result, we can show that the transition probabilities associated with the process  \eqref{process} coincide with the transition probabilities \eqref{T>condprob}, \eqref{FTcond}. From \eqref{process} we get
\[
X(t)= {S_{t-t_0}}^\T X(t_0)+\int_{t_0}^t {S_{t-\tau}}^\T\rmd B(\tau).
\]
By working with the characteristic functions and using \eqref{FTcond}, we have
\begin{multline*}
\Ebb\left[\rme^{\rmi k^\T X(t)}\Big|X(t_0)=z\right]=\exp\left\{\rmi z^\T S_{t-t_0}k\right\}\Ebb\left[\exp\left\{\rmi \int_{t_0}^t \left(S_{t-\tau}k\right)^\T\rmd B(\tau)\right\}\right]
\\ {}=\exp\left\{\rmi z^\T S_{t-t_0}k+\int_{t_0}^t\rmd \tau\, \psi_0(S_{t-\tau}k)\right\} =\exp\left\{\rmi z^\T S_{t-t_0}k+\int_{0}^{t-t_0}\rmd \tau\, \psi_0(S_{\tau}k)\right\} =\int_{\Rbb^s}\rme^{\rmi k^\T y}P_{t-t_0}(\rmd y|z);
\end{multline*}
this result  gives the identification.

Let us recall that the L\'evy process $B(t)$ can be represented in terms of Wiener processes and random Poisson measures (the L\'evy-It\^o decomposition \cite[Theor.\ 2.4.16]{App09}).

\subsubsection{The mean values}
By assuming the convergence of the integral
\[
\int_{\Rbb^s}\abs y \nu_0(\rmd y)<+\infty,
\]
we get the existence of  the mean values, which are given by
\[
\Ebb[X_j(t)]=\int_{\Rbb^s}\rmd x\, x_jp_t(x)=-\rmi \frac{\partial \hat p_t(k)}{\partial k_j}\Big|_{k=0} = \left( S_t^\T\Ebb[X(0)]\right)_j-\rmi \int_0^t \rmd \tau\, \frac{\partial \psi_0(S_\tau k)}{\partial k_j}\Big|_{k=0},
\]
\[
-\rmi \, \frac{\partial \psi_0(S_\tau k)}{\partial k_j}\Big|_{k=0}=\left({S_\tau}^\T \alpha^0\right)_j+ \int_{\abs{y}\geq 1}\nu_0(\rmd y)\left( {S_\tau}^\T y\right)_j.
\]
Then, we can write \eqref{psicl} as
\[
\psi_0(k)=\rmi\tilde\alpha^{0\,\T} k +\tilde \psi_0(k), \qquad \tilde \eta=\eta^0-\int_{\Sbb_0}y\nu_0(\rmd y), \]
\[ \tilde \psi_0(k)=-\frac 12 \,k^\T A^{00} k+\int_{\Rbb^s}\nu_0(\rmd y)\left(\rme^{\rmi y^\T k}-1\right).
\]

Now Eqs.\ \eqref{process} become
\begin{equation*}\begin{split}
&\rmd X(t)= {Z^{00}}^\T X(t)\rmd t   +\tilde\eta \rmd t+ \rmd \tilde B(t), \\ &X(t)= {S_t}^\T X(0)+\int_0^t {S_{t-\tau}}^\T\tilde \eta \rmd \tau +\int_0^t {S_{t-\tau}}^\T\rmd \tilde B(\tau);
\end{split}\end{equation*}
the increments of $\tilde B(t)$ have characteristic function
\begin{equation*}
\Ebb\left[ \exp\left\{ \rmi k^\T\left(\tilde B(t+\Delta t)-\tilde B(t)\right)\right\}\right]=\exp\left\{\tilde \psi_0(k) \Delta t\right\}.
\end{equation*}
Note that the jump contribution in $\tilde \psi_0$ is not compensated.

\subsection{An example: a dissipative harmonic oscillator}
Let us take $s=2$ and set $x(t)=X_1(t)$, \ $p(t)=X_2(t)$.  We take the deterministic part of the dynamics to be equal to the quantum case of Sec.\ \ref{sec:optomec}:  the generator of $S_t^\T$ is given by Eq.\ \eqref{lambda_pm} with ${Z^{11}}^\T\to {Z^{00}}^\T$. Then, $S_t^\T$ is given by \eqref{xi1,2}.
Moreover, in the classical case there is no restriction on the matrix $A^{00}$, apart from being non-negative definite. So, we are allowed to take $A_{11}^{00}=0$ and we can restrict also the noise and the forces to the second component (as it must be in classical equations of motion):
\begin{equation}\label{A22}
\tilde\eta_1=0, \qquad \tilde \psi_0(k)=-\frac 12  A_{22}^{00} {k_2}^2+\int_{\Rbb}m_0(\rmd v)\left(\rme^{\rmi v k_2}-1\right).
\end{equation}
Note that this gives $\tilde B_1(t)=0$ and
\begin{equation}\label{cl-a-osc}
\begin{split}
&\rmd x(t)= \Omega p(t)\rmd t,
\\
&\rmd p(t) = -\Omega x(t)\rmd t -\gamma p(t) \rmd t + \tilde\eta_2\rmd t+\rmd \tilde B_2(t).
\end{split}
\end{equation}
All the forces appear only in the second equation, while the first equation gives the meaning and the units of measure of the momentum: $p(t)=\dot x(t)/\Omega$.
We assume also the existence of the mean values; so, we have
\begin{gather*}
\frac{\rmd \langle x(t)\rangle}{\rmd t}= \Omega \langle p(t)\rangle,
\qquad
\frac{\rmd \langle p(t)\rangle}{\rmd t} = -\Omega \langle x(t)\rangle  -\gamma \langle p(t) \rangle + \tilde\eta_2+\frac {\rmd \ } {\rmd t} \langle \tilde B_2(t)\rangle,
\\
\frac {\rmd \ } {\rmd t} \langle \tilde B_2(t)\rangle=\int_{\Rbb}y\tilde\nu_0(\rmd y).
\end{gather*}
The last term is the mean pressure force due to the jump noise; if the noise is symmetric in both sides, the mean pressure force can vanish.

Similarly to \eqref{lambda_pm} and \eqref{xi1,2}, Eqs.\ \eqref{cl-a-osc} can be explicitly solved and the characteristic function \eqref{FTcond} of the transition probabilities can be computed. In particular, we get that it exists the limit for large times of the probability distribution at time $t$: from \eqref{FTpt} we get
\begin{equation*}
\lim_{t\to +\infty}\hat p_t(k)=\exp\left\{\int_0^{+\infty}\rmd \tau \left( \rmi k^\T S_\tau^\T\tilde \eta +\tilde \psi_0(S_\tau k)\right)\right\}.
\end{equation*}
This is the characteristic function of an infinitely divisible distribution. By the choice \eqref{A22} the matrix in the Gaussian component of the generating triplet \eqref{genetrpl} turns out to be
\[
A_\infty^{00}= A^{00}_{22}\int_0^{+\infty}\rmd \tau \,S_\tau^\T \begin{pmatrix}0 & 0 \\ 0 &1\end{pmatrix} S_\tau =\frac {A^{00}_{22}}{2\gamma}\begin{pmatrix}1 & 0 \\ 0 &1\end{pmatrix}.
\]
For $A^{00}_{22}>0$, we have  $\det A_\infty^{00}>0$ and the Gaussian component is not degenerate; then, the probability distribution with characteristic function $\lim_{t\to +\infty}\hat p_t(k)$ admits a density with respect to the Lebesgue measure.

\section{A generic hybrid system}\label{sec:hybrid}

The semigroup $\Tcal_t$, constructed in Theor.\ \ref{mainTheor}, acts on the $W^*$-algebra $\Nscr$ and it represents the dynamics in the Heisenberg description. Then, if $\hat\pi_0$ is the state at time $t=0$, its evolution is given by the preadjoint semigroup: $\hat\pi_t= {\Tcal_t}_*[\hat\pi_0]$. Let us recall that a generic element $F\in \Nscr$ is a function $F(x)$ from $\Rbb^s$ into $\Bscr(\Hscr)$; then,
the state $\hat\pi_t$ is a trace-class valued function $\hat\pi_t(x)\in \Tscr(\Hscr)$, $x\in \Rbb^s$, such that $\hat\pi_t(x)\geq 0$ and $\int_{\Rbb^s}\rmd x\, \Tr\{\hat\pi_t(x)\}=1$.

The characteristic function of a hybrid state has been introduced in \eqref{char+Wig}; similarly, the characteristic function of the trace-class operator $\hat\pi_t(x)$ can be defined:
\begin{equation*}\begin{split}
&\chi_{\hat\pi_t}(\xi)= \int_{\Rbb^s} \rmd y\,\Tr\left\{\hat\pi_t(y)W(\xi)(y)\right\}, \\ &\chi_{\hat\pi_t(x)}(\zeta)=  \Tr\left\{\hat\pi_t(x)W_1(\zeta)\right\}=\frac 1{(2\pi)^s}\int_{\Rbb^s}\rmd k \,\rme^{-\rmi x^\T k} \chi_{\hat\pi_t}(\zeta,k).
\end{split}\end{equation*}
By using the explicit form \eqref{f+Psi} giving the action of $\Tcal_t$ on the Weyl operators, we get
\[
\chi_{\hat\pi_t}(\xi)= \int_{\Rbb^s} \rmd x\,\Tr\left\{\hat\pi_0(x)W(S_t\xi)(x)\right\}\rme^{\Psi_t(\xi)} \\ {}= \int_{\Rbb^s} \rmd x\,\chi_{\hat\pi_0(x)}(P_1S_t\xi)\exp\left\{\Psi_t(\xi)+\rmi x^\T P_0S_t\xi\right\}.
\]

Let $X(t)$ be the random variable representing the classical component at time $t$, as in
Sec.\ \ref{sec:clprocess}. Then, we have
\begin{equation}\label{margcl}
P[X(t)\in B]=\int_{B}\rmd x\,p_t(x), \qquad p_t(x)=\Tr\left\{\hat\pi_t(x)\right\}=\chi_{\hat\pi_t(x)}(0),
\end{equation}
while the reduced quantum state $\rho_t$ is given by
\begin{equation}\label{reducstate}\begin{split}
\rho_t&=\int_{\Rbb^s} \rmd x\,\hat\pi_t(x), \\ \chi_{\rho_t}(\zeta)=\chi_{\hat\pi_t}(\zeta,0)&=\int_{\Rbb^s} \rmd x\,\chi_{\hat\pi_0(x)}(P_1S_tP_1\zeta)\exp\left\{\Psi_t(\zeta,0)+\rmi x^\T P_0S_tP_1\zeta\right\}.
\end{split}\end{equation}
The  expressions \eqref{margcl} and \eqref{reducstate} can be seen as the marginals of the classical/quantum state $\hat\pi_t$.

In the pure classical case of Sec.\ \ref{sec:clprobs}, we have seen that the dynamical semigroup gives rise not only to the probabilities at time $t$ (the state at time $t$), but also to the joint probabilities at different times and to a whole stochastic process in continuous time.
In a quantum theory probabilities and state changes are unified in the notion of \emph{instrument} \cite[Sec.\ 4.1]{Hol01}, \cite[Sec.\ B.4]{BarG09}, \cite[Sec.\ 1.4]{WisM10}. As in the pure classical case transition probabilities are needed (see Sec.\ \ref{sec:conprob}) to construct the related process, in the hybrid system we need to introduce a kind of ``transition'' instruments \cite{BarH95}.

\subsection{Instruments and probabilities}\label{sec:I+prob}
In the theory of measurements in continuous time, the classical component is the output of the measurement (usually without a proper dynamics and backaction on the quantum system) \cite{BarL91,BarHL93,BarP96}; in this situation,   the connection between the semigroup $\Tcal_t$ and the instruments is given by the equation \cite[(2.5)]{BarHL93}. In the classical case we have the transition probabilities \eqref{T>condprob}; by analogy, we need to introduce instruments depending on the initial value of the classical system, some kind of \emph{transition instruments} as done in \cite[Sec.\ 4.4]{BarH95}. So, we define the family of instruments
\begin{equation}\label{I|B}
\Ical_t(E|x)[a]=\Tcal_t[a\otimes \ind_E](x), \qquad \forall a\in\Bscr(\Hscr), \qquad \forall x\in \Rbb^s;
\end{equation}
$E\subset \Rbb^s$ is a generic Borel set.

\begin{proposition}\label{prop:transinstr}
Equation \eqref{I|B} actually defines an instrument $\Ical_t(\cdot|x)$ on the $\sigma$-algebra of the Borel sets in $\Rbb^s$, which means that the following properties hold:
\begin{enumerate}
\item for every Borel set $E\subset \Rbb^s$ and  $x\in \Rbb^s$ (almost everywhere with respect to the Lebesgue measure), $\Ical_t(E|x)$ is a completely positive and normal map from $\Bscr(\Hscr)$ into itself;
\item (normalization) $\Ical_t(\Rbb^s|x)[\openone]=\openone$;
\item ($\sigma$-additivity) for every countable family of Borel disjoint sets $E_i$, $\Ical_t\left(\bigcup_i E_i\big| x\right)[a]=\sum_i\Ical_t(E_i|x)[a]$, $\forall a\in \Bscr(\Hscr)$.
\end{enumerate}

Moreover, the family of instruments \eqref{I|B} enjoys  the following composition property:
\begin{equation}\label{qCKeq}
\Ical_{t+t'}(E|x)=\int_{z\in \Rbb^s} \Ical_{t'}(\rmd z|x)\circ\Ical_t(E|z).
\end{equation}

Finally,  the action of the transition instrument on the Weyl operators is given by
\begin{equation}\label{trans+Weyl}
\Ical_t(E|x)[W_1(\zeta)]= \frac 1 {(2\pi)^s}\int_{\Rbb^s}\rmd k \int_E \rmd z \,\exp\left\{ \Psi_t(\xi)-\rmi z^\T k\right\}W(S_t\xi)(x), \quad \xi=\begin{pmatrix} \zeta \\ k \end{pmatrix}.
\end{equation}
\end{proposition}
Properties 1.--3.\ in the proposition above come from the definition of \emph{instrument}. Equation \eqref{qCKeq} is the quantum analogue of the Chapman-Kolmogorov identity \eqref{CKeq}.

\begin{proof}
By using the defining equation \eqref{I|B} and the properties of $\Tcal_t$ in Definition \ref{def:new1}, we have that
\begin{itemize}
\item point a. implies property 1.,
\item point b. implies property 2.,
\item again normality (point a.), linearity, and the structure with the indicator function imply property 3.
\end{itemize}

For any $F(\cdot)\in \Nscr$ we have $\int_{\Rbb^s}\Ical_t(\rmd y|x)[F(y)]=\Tcal_t[F](x)$. Then, we have
\begin{multline*}
\int_{\Rbb^s}\Ical_{t+t'}(\rmd y|x)[F(y)]= \Tcal_{t+t'}[F](x)=\Tcal_{t'}\circ\Tcal_t[F](x)
\\ {}= \int_{y\in \Rbb^s} \Tcal_{t'}\left[ \Ical_t(\rmd y|\cdot)[F(y)]\right](x)
= \int_{y\in \Rbb^s} \int_{z\in \Rbb^s} \Ical_{t'}(\rmd z|x)\circ \Ical_t(\rmd y|z)[F(y)];
\end{multline*}
this gives \eqref{qCKeq}.

By using $\ind_{E}(x)= \frac 1 {(2\pi)^s}\int_{\Rbb^s}\rmd k \int_{E} \rmd z \, \rme^{\rmi  (x-z)^\T k}$, we get by easy computations
\begin{equation*}
\Ical_t(E|x)[W_1(\zeta)]=\frac 1 {(2\pi)^s}\int_{\Rbb^s}\rmd k \int_{E} \rmd z \, \rme^{-\rmi k^\T z}\Tcal_t[W(\xi)](x);
\end{equation*}
by \eqref{Tqf} and \eqref{f+Psi}, Eq.\ \eqref{trans+Weyl} follows.
\end{proof}

From \eqref{trans+Weyl} we obtain also the explicit form of the Fourier transform of the transition instrument, the \emph{characteristic operator}:
\begin{equation}\label{Gamma1}
\Gamma_t(k|x)[W_1(\zeta)]:=\int_{\Rbb^s} \Ical_t(\rmd z|x)[W_1(\zeta)]\rme^{\rmi k^\T z}= \Tcal_t[W(\xi)](x) =\exp\left\{ \Psi_t(\xi)\right\}W(S_t\xi)(x), \qquad \xi=\begin{pmatrix} \zeta \\ k \end{pmatrix}.
\end{equation}

\subsubsection{Multi-time probabilities and instruments}
By repeated application of the single-time instrument, also multi-time instruments can be obtained, again determined by their Fourier transform. Let us take the arbitrary times $0< t_1<t_2<\cdots<t_m$; the multi-time analogue of \eqref{Gamma1} is the quantity
\begin{multline*}
\Gamma(k_1,t_1;\ldots;k_m,t_m|\zeta, x)= \int_{\left(\Rbb^s\right)^m}\Ical_{t_1}(\rmd x_1|x)\circ \Ical_{t_2-t_1}(\rmd x_2|x_1)\circ \cdots \\ {}\circ \Ical_{t_m-t_{m-1}}(\rmd x_m|x_{m-1})[W_1(\zeta)]\exp\left\{\sum_{l=1}^m \rmi x_l^\T k_l\right\}.
\end{multline*}
By recursion, from \eqref{Tqf}, \eqref{f+Psi}, and \eqref{I|B}, we get
\begin{equation*}\begin{split}
&\Gamma(k_1,t_1;\ldots;k_m,t_m|\zeta, x)= \exp\left\{\sum_{l=1}^m \Psi_{t_l-t_{l-1}}(\xi_l)+\rmi x^\T P_0S_{t_1} \xi_1\right\}W_1(P_1S_{t_1} \xi_1),
\\
& t_0=0, \qquad \xi_l=S_{t_{l+1}-t_l} \xi_{l+1}+P_0k_l, \quad l=1,\ldots,m-1, \qquad \xi_m=\begin{pmatrix}\zeta \\ k_m\end{pmatrix}.
\end{split}
\end{equation*}

Condition \eqref{qCKeq} implies the consistency of the multi-time probabilities and the existence of a stochastic process in continuous time as in the purely classical case (Sec.\ \ref{sec:conprob}).
The joint probabilities for the process $X(t)$ at the times $0< t_1<t_2<\cdots<t_m$ (for an initial state $\hat\pi_0$) are given by
\begin{multline}\label{genprobs}
P[X(t_1)\in E_1, X(t_2)\in E_2, \ldots, X(t_m)\in E_m|\hat\pi_0]
\\ {}=\int_{\Rbb^s}\rmd x \Tr \biggl\{\hat\pi_0(x) \int_{E_1}\Ical_{t_1}(\rmd x_1|x)\circ \int_{E_2}\Ical_{t_2-t_1}(\rmd x_2|x_1)\circ \cdots \circ \int_{E_m}\Ical_{t_m-t_{m-1}}(\rmd x_m|x_{m-1})[\openone]\biggr\}.
\end{multline}
Then, the quantity
$\int_{\Rbb^s}\rmd x \Tr \left\{\hat\pi_0(x)\Gamma(k_1,t_1;k_2,t_2;\ldots;k_m,t_m|0, x)\right\}$ turns out to be the characteristic function of these multi-time probabilities. To better understand this result let us consider the two times case:
$0<t_1<t_2$. By simple computations we get the characteristic function
\begin{multline}\label{2ts}
\Ebb\left[\exp\left\{\rmi k_1^\T X(t_1)+\rmi k_2^\T X(t_2)\right\}\right]=\int_{\Rbb^s}\rmd x \Tr \left\{\hat\pi_0(x)\Gamma(k_1,t_1;k_2,t_2|0, x)\right\}
\\ {}=\int_{\Rbb^s}\rmd x \Tr \left\{\hat\pi_0(x)W_1(P_1S_{t_1}\xi_1)\right\}\exp\left\{\rmi x^\T P_0S_{t_1}\xi_1+\Psi_{t_1}(\xi_1) +\Psi_{t_2-t_1}(\xi_2)\right\},
\end{multline}
where
\[
\xi_2=P_0k_2, \qquad \xi_1= S_{t_2-t_1}P_0k_2+P_0k_1, \qquad S_{t_1}\xi_1= S_{t_1}P_0k_1+ S_{t_2}P_0k_2,
\]
\[
\Psi_{t_2-t_1}(\xi_2)=\int_0^{t_2-t_1}\rmd \tau \psi(S_\tau P_0 k_2), \quad \Psi_{t_1}(\xi_1)=\int_0^{t_1}\rmd \tau \psi(S_\tau P_0 k_1+S_{t_2-t_1+\tau}P_0k_2).
\]

\subsubsection{Conditional probabilities and conditional states}
Having constructed multi-time probabilities \eqref{genprobs} and instruments \eqref{I|B}, it is possible to introduce both probabilities and quantum states conditioned on the previous observations. Let us exemplify the introduction of conditional probabilities and states in a simple case: a single application of the transition instrument. The role of the notion of instrument in a quantum theory is to give the probabilities for the associated measurement and the state after the measurement conditioned on the observed result. In a time interval $(0,t)$, the probability of observing the increment $X(t)-X(0)$ of the classical subsystem in a Borel set $E\subset \Xi_0$, conditional on the initial value $X(0)=x$ and on the state $\rho_0$ of the quantum subsystem, is given by
\begin{equation}\label{cprob}
P_{t}(E|x):= P\big(X( t)-X(0)\in E|X(0)=x,\rho_0\big) =\Tr\left\{\rho_0 \Ical_{ t}(E|x)[\openone]\right\}.
\end{equation}
Then, the quantum state at time $t$, conditional on the observation of the classical component in $E$, is
\begin{equation*}
\rho_{t}(E;x)=\frac{\Ical_{ t}(E|x)_*[\rho_0]}{P_{ t}(E|x)};
\end{equation*}
also the limit case of the set $E$ reducing to a point $z$ can be considered. By using the characteristic function for the quantum state we get
\[
\chi_{\rho_{ t}(E;x)}(\zeta)=\Tr\left\{\rho_{ t}(E;x)W_1(\zeta)\right\} =\frac{\Tr\left\{\rho_0 \Ical_{ t}(E|x)[W_1(\zeta)]\right\}}{P_{ t}(E|x)}.
\]

In the quasi-free case, using $\xi^\T=(\zeta^\T,k^\T)$, we have
\[
\Tr\left\{\rho_0 \Ical_{ t}(E|x)[W_1(\zeta)]\right\}= \frac 1 {(2\pi)^s}\int_{\Rbb^s}\rmd k \int_E \rmd z \,\exp\left\{ \Psi_{ t}(\xi)-\rmi z^\T k+ \rmi x^\T P_0S_{ t}\xi\right\}\chi_{\rho_0}(P_1S_{ t}\xi).
\]
In particular, the probabilities \eqref{cprob} become
\begin{equation*}
P_{t}(E|x)= \frac 1 {(2\pi)^s}\int_{\Rbb^s}\rmd k \int_E \rmd z \,\exp\left\{ \Psi_{ t}(P_0 k)-\rmi z^\T k+ \rmi x^\T P_0S_{ t}P_0 k\right\}\chi_{\rho_0}(P_1S_{ t}P_0k);
\end{equation*}
the quantum state $\rho_0$ affects these probabilities only through $P_1S_{ t}P_0$. For a small time interval $\Delta t$, we have $P_1S_{\Delta t}P_0$ $\simeq Z^{10} \Delta t$ and $\chi_{\rho_0}(P_1S_{\Delta  t}P_0k)\simeq \exp\left\{\rmi \Delta t \langle R^\T\rangle_0 P_1 Z^{10}k\right\}$. By comparing with \eqref{K12int2} we see once more that the classical probabilities acquire information from the quantum system via the interaction term $\Kcal^2_{\rm int}$. Moreover, in the approximation for small times, we have that $P_{\Delta t}(\bullet|x)$ is an infinitely divisible distribution, because we have
\[
P_{\Delta t}(E|x)\simeq \frac 1 {(2\pi)^s}\int_{\Rbb^s}\rmd k \int_E \rmd z \,\exp\left\{ \Delta t\psi(P_0 k)+\rmi \left(x-z\right)^\T k+ \rmi  \Delta t x^\T  Z^{00} k + \rmi \Delta t \langle R^\T\rangle_0 P_1 Z^{10}k\right\}.
\]

\subsection{An example}\label{sec:exhyb}

To illustrate the properties of hybrid systems and the role of the interaction terms discussed in Sec.\ \ref{sec:gen+int}, we present a simple example with a classical system injecting noise in the quantum component and an observed output signal; as quantum system we take the quantum mechanical oscillator introduced in Sec.\ \ref{sec:optomec}. So, we take $n=1$, $s=2$; as generator of the deterministic evolution $S_t=\rme^{Z t}$ we take
\begin{equation*}
Z=\begin{pmatrix} 0&-\Omega&0 & g\\ \Omega & -\gamma & 0 & 0\\ 0 & b & -c& 0\\ 0& 0& 0& 0\end{pmatrix},  \qquad \begin{matrix} \Omega>0 , \qquad  \gamma >0, \qquad  c>0, \qquad  b,\,g\in \Rbb ,
\\
\phantom{n}
\\
\lambda_\pm=-\frac \gamma 2 \pm \rmi \omega, \qquad \ \omega=\sqrt{\Omega^2-\frac{\gamma^2}4}>0.
\end{matrix}
\end{equation*}
By setting $\xi(t)=S_t\xi(0)$ we can easily solve the first order equations $\dot\xi(t)=Z\xi(t)$; under the initial condition $\xi(0)= (\zeta_1 ,\, \zeta_2 ,\, k_1 ,\, k_2)^\T$, we get
\begin{equation*} \begin{cases}
\xi_1(t)=x_+\rme^{\lambda_+ t} + x_- \rme^{\lambda_- t}+\frac{\gamma gk_2}{\Omega^2}\,,
\\
\xi_2(t)=y_+\rme^{\lambda_+ t} + y_- \rme^{\lambda_- t}+\frac{ gk_2}{\Omega}\,,
\\
\xi_3(t)=\rme^{-ct}k_1 +b\ell(t),
\\
\xi_4(t)=k_2\,,
\end{cases}\end{equation*}
\begin{equation*} \begin{cases}
\ell(t)= y_+\frac{\rme^{\lambda_+t}-\rme^{-ct}}{\lambda_++c}+ y_-\frac{\rme^{\lambda_-t}-\rme^{-ct}}{\lambda_-+c} + gk_2\frac{1-\rme^{-ct}}{\Omega c} \,,
\\
2\rmi \omega x_+=-\lambda_-\zeta_1-\Omega \zeta_2- gk_2,
\\
2\rmi \omega y_+=\Omega \zeta_1+\lambda_+\left(\zeta_2+\frac{ gk_2}{\Omega}\right),
\\
x_-= \overline{x_+},\qquad y_-= \overline{y_+}.
\end{cases}
\end{equation*}

By \eqref{B,Z}, the form of $Z$ implies
\begin{equation*}
B=\frac{\sigma Z- Z^\T\sigma^\T}2=\frac 12 \begin{pmatrix} 0 &-\gamma &0& 0 \\ \gamma & 0 & 0&-g \\ 0 & 0 & 0 &0\\ 0 & g&0&0\end{pmatrix}, \qquad A-\rmi B\geq 0.
\end{equation*}
We take a very simple choice for $A$, compatible with the positivity condition above:
\begin{equation*}
A=\begin{pmatrix} A_{11} & A_{12} &0& 0 \\ A_{12} & A_{22} & 0&A_{24} \\ 0 & 0 & A_{33} &0\\ 0 &A_{24}&0&A_{44}\end{pmatrix}, \qquad \begin{matrix} A_{ii}\geq 0, \qquad & A_{22}=A_{22}^1+A_{22}^2,
\\
A_{22}^1\geq 0, \qquad  & A_{11}A_{22}^1\geq {A_{12}^{2}+}\frac {\gamma^2}4\,,
\\
A_{22}^2\geq 0,\qquad & A^2_{22}A_{44}\geq {A_{24}}^2+\frac {g^2}4\,.
\end{matrix}
\end{equation*}
The action of the dynamical semigroup on the Weyl operators is given by \eqref{Tqf} and \eqref{f+Psi}, where $\psi(\xi)$ has the expression \eqref{formpsi}. We choose a simpler expression also for this quantity, with the jump part only in the first component of the classical system:
\begin{equation}\label{exe:psi}
\psi(\xi)=-\frac 12 \,\xi^\T A \xi+\rmi \alpha^0 \xi_3 +J_3(\xi_3),\qquad  J_3(\xi_3)= \int_{\Rbb\setminus\{0\}}\nu(\rmd v)\left(\rme^{\rmi v \xi_3}-1-\rmi \ind_{\{\abs{v}<1\}}(v)\, v\xi_3\right) .
\end{equation}
To have the limit for large times, according to the discussion in Sec.\ \ref{sec:equil}, we add the assumption
\[
\int_{\abs v >1}\ln \abs v \,\nu(\rmd v)<+\infty.
\]
With these choices, the quantum/classical interaction terms in \eqref{interactions} become
\begin{subequations}\label{ex:interac}
\begin{equation}
\Kcal_{\rm int}^1[a\otimes f](x)=\rmi \left[H_x,\, a\right]f(x)=-\rmi b x_1 \left[Q,\, a\right]f(x),\end{equation}
\begin{equation} \Kcal_{\rm int}^2[a\otimes f](x)= \frac g 2 \{a,Q\} \,\frac{\partial f(x)}{\partial x_2},
\end{equation}
\begin{equation}
\Kcal_{\rm int}^3[a\otimes f](x)=-  \rmi[Q,a ] A_{24}\,\frac{\partial f(x)}{\partial x_2}, \qquad \Kcal_{\rm int}^4=0.
\end{equation}
\end{subequations}

\subsubsection{The input classical noise}
Let us start by considering only the input classical noise $X_1(t)$, which means to take $\zeta_1=0$, $\zeta_2=0$, $k_2=0$. This gives $x_\pm=y_\pm=0$, \ $\xi_1(t)=\xi_2(t)=\xi_4(t)=0$, \ $\xi_3(t)= \rme^{-c t}k_1$,
\begin{equation*}
\Tcal_t[W(0,0,k_1,0)](x)=\exp\left\{\int_0^t\rmd \tau \left(\rmi \alpha^0 k_1 \rme^{-c\tau} -\frac {A_{33}}2 \, k_1^{\;2}\rme^{-2c\tau}+J_3\big( k_1\rme^{-c\tau}\big)\right) \right\} \openone.
\end{equation*}
By this, we see that the quantum system is not involved in the evolution of the classical process $X_1(t)$. So, we can directly apply the results of Sec.\ \ref{sec:clprocess}; by  \eqref{process} and \eqref{Lprocess}, $X_1(t)$ is the process
\begin{equation*}
X_1(t)= \rme^{-c t} X(0)+\int_0^t \rme^{-c(t-\tau)}\rmd B_1(\tau);
\end{equation*}
$B_1(t)$ is a L\'evy process, whose increments have characteristic function
\begin{equation*}
\Ebb\left[ \exp\left\{ \rmi k_1\bigl(B_1(t+\Delta )-B_1(t)\bigr)\right\}\right]=\exp\left\{\Delta \left(\rmi \alpha^0 k_1  -\frac {A_{33}}2 \, k_1^{\;2}+J_3\big( k_1\big)\right)\right\}.
\end{equation*}

\subsubsection{The reduced quantum state}
The reduced quantum state does not satisfy a Markovian master equation, but, being the evolution quasi-free, we have implicitly the expression of the reduced state at any time. To see the noise acquired from the classical component, we study the quantum state at large times. We take $k_1=0$, $k_2=0$, which gives
\begin{subequations}\label{xiII}
\begin{gather}
\begin{cases}
\xi_1(t)=x_+\rme^{\lambda_+ t} + x_- \rme^{\lambda_- t},
\\
\xi_2(t)=y_+\rme^{\lambda_+ t} + y_- \rme^{\lambda_- t},
\\
\xi_3(t)=b\left( y_+\frac{\rme^{\lambda_+t}-\rme^{-ct}}{\lambda_++c}+ y_-\frac{\rme^{\lambda_-t}-\rme^{-ct}}{\lambda_-+c} \right),
\\
\xi_4(t)=0,
\end{cases} \\ \begin{cases}
{\displaystyle x_+=\overline{x_-}=\frac12\,\zeta_1-\frac\rmi 2 \left(\frac \gamma{2\omega}\,\zeta_1-\frac \Omega \omega\,\zeta_2\right)=-\frac{\lambda_-y_+}\Omega,}
\\
{\displaystyle y_+=\overline{y_-}=\frac12\,\zeta_2+\frac\rmi 2 \left(\frac \gamma{2\omega}\,\zeta_2-\frac \Omega \omega\,\zeta_1\right) =-\frac{\lambda_+x_+}\Omega;}
\end{cases}
\end{gather}
\end{subequations}
by the choice of the quantum system, $\xi_1(t)$ and $\xi_2(t)$ have the expression \eqref{xi1,2} of the example of Sec.\ \ref{sec:optomec}. The reduced quantum state is given by \eqref{reducstate}. The expressions for $\xi_j(t) $ in \eqref{xiII} decay exponentially; then,
the quantum characteristic function at large times turns out to be
\[
\chi_{\rho_{\rm eq}(\zeta)=}\lim_{t\to +\infty}\chi_{\rho_t}(\zeta)=\exp\left\{\int_0^{+\infty}\rmd \tau \,\psi\big(\xi(\tau)\big)\right\}.
\]
By \eqref{exe:psi} and \eqref{xiII}, we get
\[
\int_0^{+\infty}\rmd \tau \,\psi\big(\xi(\tau)\big)= \rmi\frac{ b\alpha^0\zeta_2}{c\Omega}-\frac 12 \,\zeta^\T A^{\rm q} \zeta + \int_0^{+\infty}\rmd \tau \,J_3\big(\xi_3(\tau)\big),
\]
\[
A^{\rm q}_{11}=A^{\rm q}_{22}+ \frac{A_{12}}{\Omega}  , \qquad A^{\rm q}_{12}=A^{\rm q}_{21}=0,
\qquad
A^{\rm q}_{22}=\frac{A_{11}+A_{22}}{2\gamma} +  \frac{A_{33} b^2 \left(\frac \gamma 2 +c\right)}{2c\gamma\left(\Omega^2+\gamma c+ c^2\right)}.
\]
Without the interaction with the first component $X_1(t)$ of the classical system ($b=0$, which in this example means $\Kcal^1_{\rm int}=0$), the asymptotic state of the quantum system becomes purely Gaussian, as the jump integral $J_3$ vanishes together with the terms proportional to $A_{33}$. The interaction of the quantum oscillator with $X_1(t)$ modifies the Gaussian component of the state, by adding the terms proportional to $A_{33}$. Moreover, while this interaction is linear (the first expression in \eqref{ex:interac}), the asymptotic quantum state is not Gaussian, but it contains also noise of ``jump type'' received from the classical noise $X_1(t)$.

\subsubsection{The observed output}
We consider now the observed output, the classical component $X_2(t)$ alone. Without interaction, i.e.\ $g=0$, $A_{14}=0$, only the term with $A_{44} $ survives and $X_2(t)$ is a Wiener process with variance $t A_{44}$. For the case with interaction, we discuss only some of the involved  probabilities, not the full process. By taking $\zeta_1=0$, $\zeta_2=0$, $k_1=0$, $k_2=\varkappa$, we get
\begin{equation*} \begin{cases}
\xi_1(t;\varkappa)=\varkappa g\left(\frac{\rme^{\lambda_- t}-\rme^{\lambda_+ t}}{2\rmi \omega}+ \frac \gamma{\Omega^2}\right),
\\
\xi_2(t;\varkappa)=\varkappa\,\frac{ g}{\Omega}\left( \frac{\lambda_+\rme^{\lambda_+ t}-\lambda_-\rme^{\lambda_- t}}{2\rmi \omega}\, +1\right),
\\
\xi_3(t;\varkappa)=\varkappa\,\frac{bg}{\Omega}\left(2\RE \frac{\lambda_+\left(\rme^{\lambda_+t}-\rme^{-ct}\right)}{2\rmi\omega\left(\lambda_++c\right)} +\frac{1-\rme^{-ct}}{ c} \right),
\\
\xi_4(t;\varkappa)=\varkappa.
\end{cases}
\end{equation*}
By using \eqref{Gamma1} we obtain the characteristic function of the one-time probability law
\[
\Ebb\left[\rme^{\rmi \varkappa X_2(t)}\right]= \int_{\Rbb^2}\rmd x \Tr\left\{\hat\pi_0(x)W_1\big(P_1\xi(t;\varkappa)\big) \right\}
\exp\left\{\rmi x^\T P_0\xi(t;\varkappa)+\int_0^t\rmd \tau \, \psi\big(\xi(\tau;\varkappa)\big)\right\}.
\]
For large times we have
\begin{gather*}
\xi(t;\varkappa)\simeq \xi(\infty;\varkappa)=\varkappa\left(\frac{ g\gamma }{\Omega^2} \,, \ \ \frac{ g}{\Omega}\,, \ \ \frac{bg}{c\Omega}\,, \ \ 1\right)^\T,
\\
\int_0^t\rmd \tau \, \psi\big(\xi(\tau;\varkappa)\big)\simeq \int_0^{+\infty}\rmd \tau \left[ \psi\big(\xi(\tau;\varkappa)\big)-\psi\big(\xi(\infty;\varkappa)\big)\right]+ t\psi\big(\xi(\infty;\varkappa)\big),
\end{gather*}
\[
\psi\big(\xi(\infty;\varkappa)\big)= J_3\Big(\frac {bg}{c\Omega}\Big)+ \rmi \frac{bg}{c\Omega}\,\alpha^0 \varkappa -\frac {\varkappa^2}2 \left(\frac {g^2\gamma^2 A_{11}}{\Omega^4}+\frac {g^2A_{22}}{\Omega^2} +\frac{2g^2\gamma  A_{12}}{\Omega^3}+ 2\frac {gA_{24}}{\Omega}+A_{44} + \frac {b^2g^2A_{33}}{c^2\Omega^2}\right).
\]
Obviously, we have to assume the existence of the integral in the central expression.
These expressions show that the process $X_2(t)$ cumulates contributions coming from the quantum system (the terms with the coupling constants $g$ and $A_{14}$), but also from the classical component $X_1(t)$; however, these last contributions have to pass through the quantum system as shown by the terms with the product $bg$ of coupling constants.

To see the contribution of the interaction terms in a short time, we consider the distribution of the increment for unit of time: $\Delta^{-1}\bigl(X_2(t+\Delta)-X_2(t)\bigr)$, with $\Delta$ small. It is not possible to take the limit of vanishing $\Delta$ because $X_2(t)$ contains a Wiener component, whose derivative is a white noise (a generalized stochastic process). In the two-time probability \eqref{2ts} we take $t_1=t$, \ $t_2=t+\Delta$, \ $k_2=\frac\varkappa\Delta (0, 1)^\T$, \ $k_1=-\frac\varkappa\Delta (0, 1)^\T$, and we get, for small $\Delta$,
\begin{gather*}
\xi_1= \frac\varkappa \Delta \left(S_\Delta-1\right) P_0 \begin{pmatrix} 0 \\ 1 \end{pmatrix}\simeq \varkappa ZP_0 \begin{pmatrix} 0 \\ 1 \end{pmatrix} = \varkappa gP_1\begin{pmatrix} 1 \\ 0  \end{pmatrix},
\\
\xi_2=\frac\varkappa \Delta\,P_0 \begin{pmatrix} 0 \\ 1 \end{pmatrix}, \qquad S_{t}\xi_1\simeq g\varkappa Y(t), \qquad \Psi_{t}(\xi_1)= \int_0^{t}\rmd \tau \,\psi\left(g\varkappa Y(\tau)\right),
\\
Y(t)= \frac{1}{2\rmi \omega} \begin{pmatrix}  \lambda_+\rme^{\lambda_-t}- \lambda_-\rme^{\lambda_+t} \\ \Omega\left(\rme^{\lambda_+t}- \rme^{\lambda_-t}\right) \\ b\Omega\left(\frac{\rme^{\lambda_+t} - \rme^{-ct}} {\lambda_+ +c} -\frac{\rme^{\lambda_-t} - \rme^{-ct}} {\lambda_- +c} \right) \\ 0 \end{pmatrix},
\end{gather*}
\begin{multline*}
\Psi_{\Delta}(\xi_2)\simeq \Delta\int_0^{1}\rmd x \,\psi\left(\frac{\varkappa}\Delta\, S_{x\Delta} P_0\begin{pmatrix} 0\\ 1\end{pmatrix}\right)\simeq \Delta\int_0^{1}\rmd x \,\psi\left(\varkappa\left(\frac{1}\Delta\, +xZ\right) P_0\begin{pmatrix} 0\\ 1\end{pmatrix}\right)
\\
= \Delta\int_0^{1}\rmd x \,\psi\left(\frac{\varkappa}\Delta\,  P_0\begin{pmatrix} 0\\ 1\end{pmatrix}+ \varkappa g xP_1\begin{pmatrix} 1\\ 0\end{pmatrix}\right)  =-\frac{A_{44}}{2\Delta}\,\varkappa^2  -\frac{A_{11}}{6}\,\Delta g \varkappa^2.
\end{multline*}
We also change the initial time from 0 to $t_0$, which means $\hat \pi_0 \to \hat \pi_{t_0}$. Then \eqref{2ts}  becomes
\begin{multline}\label{CFincrements}
\Ebb\left[\exp\left\{\rmi \varkappa\,\frac{ X_2(t_0+t+\Delta)- X_2(t_0+t)}\Delta\right\}\right]
\\ {}\simeq\int_{\Rbb^2}\rmd x \Tr \left\{\hat\pi_{t_0}(x)W_1(g\varkappa P_1Y(t))\right\}
\exp\left\{\rmi x_1 gY_3(t)\varkappa  -\frac{A_{44}}{2\Delta}\,\varkappa^2+ \int_0^{t}\rmd \tau \,\psi\big(g\varkappa Y(\tau)\big)\right\}.
\end{multline}
Note that there is no dependence on $x_2$ inside the exponential which means that these probabilities do not depend on the distribution of $X_2(t_0)$. For large $t$, $Y(t)$ vanishes and any memory of  the state $\hat \pi_{t_0}$  is lost. For $t=0$ we get $Y_i(0)=\delta_{i1}$ and \eqref{CFincrements} becomes
\begin{equation*}
\Ebb\left[\exp\left\{\rmi \varkappa\,\frac{ X_2(t_0+\Delta)- X_2(t_0)}\Delta\right\}\right]
\simeq \Tr \left\{\rho_{t_0}\rme^{\rmi Q g \varkappa}\right\}
\exp\left\{  -\frac{A_{44}}{2\Delta}\,\varkappa^2\right\}, \qquad \rho_{t_0}=\int_{\Rbb^2}\rmd x\,\hat \pi_{t_0}(x).
\end{equation*}
This means that the distribution of the increment per unit of time is the quantum distribution of the position in the state $\rho_{t_0}$ smoothed by a large Gaussian. The state $\rho_{t_0}$ depends on the whole previous history, and it could be the conditional state with respect to previous observations (an argument not developed here).

\section{Conclusions}\label{sec:concl}
In this article we have afforded the problem of giving a unified treatment of interacting classical and quantum systems. The formalism we adopted allows to unify quantum master equations, Liouville equation, Kolmogorov-Fokker-Planck equation. The treatment has been restricted to the case of a Markovian and quasi-free dynamics. The notion of quasi-free dynamics, Sec.\ \ref{Sec:qfSem}, generalizes the well known Gaussian case by including contributions of ``jump type''. The problem of finding the most general Markovian quasi-free dynamics has been solved, see Theor.\ \ref{mainTheor}; in the construction, connections with the classical \LevKhi\  formula have been used, Sec.\ \ref{sec:step3}.

The restriction to the quasi-free case allows to give explicitly the dynamical map, without the need of solving the evolution equations, involving unbounded generators. However, also the formal structure of the generator is given (Sec.\ \ref{sec:generator}), because this gives hints to construct hybrid dynamics out of the quasi-free case. As a byproduct, we have obtained the most general quasi-free quantum dynamical semigroup, together with the structure of its generator, and we have presented examples which underline the role of the ``jump'' contributions (Sec.\ \ref{sec:qds}).

The results in the pure classical case (Sec.\ \ref{classical}) are all essentially known. However, this case allows to introduce the idea that not only probabilities at one time are involved in the theory, but also transition probabilities, multi-time probabilities, probabilities for stochastic processes in continuous time. In the general case (Sec.\ \ref{sec:hybrid}), these ideas are translated in the introduction of \emph{transition instruments} (Sec.\ \ref{sec:I+prob}), which allow for the construction of multi-time probabilities for the classical component. By analyzing the generator of the hybrid dynamics (Sec.\ \ref{sec:gen+int}), it is possible to identify the interactions responsible of the action of the classical system on the quantum one and the interactions responsible of the flux of information from the quantum system to the classical one. The classical component can be observed, in principle without disturbance, and information about the quantum system can be collected. This last feature connects dynamical semigroups for hybrid systems to the theory of quantum measurements in continuous time. Let us stress that the whole construction respects the general idea that no information can be extracted from a quantum system without dissipation; indeed, by asking the vanishing of the dissipative terms in the quantum part of the generator, we get that also the terms responsible of the flux of information towards the classical component have to vanish. Finally, in Sec.\ \ref{sec:exhyb}, an example is elaborated, which shows how a classical system can inject noise of Gaussian and jump type into the quantum one, and how it can extract information from the quantum component.

Many points are left open for future developments. In the theory of open systems, the notion of \emph{stochastic dilation} or \emph{stochastic unraveling} has been introduced, to mean the case of quantum master equations written in form of (classical) stochastic differential equations \cite{BarG09,WisM10,Hol96}. This can be done also for the hybrid case, as shown in \cite{Opp+23,BarH95}, but not all possibilities have been explored. Another old idea for hybrid systems is to treat the classical component as quantum, but with only commuting operators appearing as observables. The connections of L\'evy processes with Bose quantum fields and \emph{quantum stochastic calculus} have been worked out in detail \cite{Parthas92,Parthas18}. By using these means, in the case of the monitoring in time of a quantum system, the observed signal has been dilated to a quantum system and represented by quantum observables, commuting among themselves also at differen times in the Heisenberg picture \cite{Bar86,Bar06}. The same construction for the dynamics of general hybrid systems is to be done, as it is open the problem of not quasi-free hybrid systems, possibly with unbounded generators \cite{BarC11,SieHW17}.

Inside the theory of quantum measurements in continuous time, also the notions of \emph{quantum filtering}, \emph{quantum trajectories}, \emph{conditional states} have been developed \cite{BarG09,WisM10,Bar06,ZolG97,BarB91,Bel88,Bel89}; the introduction of these notions for hybrid systems is clearly possible and we expect this to be fruitful. The classical component of the hybrid can be observed, or partially observed, and the whole filtering theory for stochastic processes enters into play. As we wrote, in the hybrid case there is transfer of information also from the classical to the quantum component and it could be interesting to connect this to the idea of \emph{feedback} on quantum systems \cite{WisM10}. Obviously, the final open problem is to find physical application of the new ideas on the dynamics of hybrid systems in quantum optics and quantum information, beyond the quasi-free case.

\end{document}